\documentclass[article, a4]{amsart}
\usepackage {amsmath, amssymb, color, textcomp, amsfonts, mathtools}
\usepackage{mathtools}

\usepackage[top=0.8in, bottom=0.9in, left=0.9in, right=0.9in]{geometry}

\makeatletter
\newtheorem{thm}{Theorem}[section]
\newtheorem{cor}[thm]{Corollary}
\newtheorem{lem}[thm]{Lemma}
\newtheorem{defn}[thm]{Definition}
\newtheorem{prop}[thm]{Proposition}

\newtheorem{preremark}[thm]{Remark}
\newenvironment{remark}%
  {\begin{preremark}\upshape}{\end{preremark}}
\newtheorem{preexample}[thm]{Example}
\newenvironment{example}%
  {\begin{preexample}\upshape}{\end{preexample}}
\newtheorem{notation}[thm]{Notation}
\numberwithin{equation}{section}

\numberwithin{equation}{section}

\numberwithin{thm}{section}
\newtheorem{lem/defn}[thm]{Lemma/Definition}
\newtheorem{preex/defn}[thm]{Example/Definition}
\newenvironment{ex/defn}%
  {\begin{preex/defn}\upshape}{\end{preex/defn}}

\newcommand{\ten}{\otimes}

\begin{document}

\title[The two bosonizations of the CKP hierarchy: overview and character identities]{The two bosonizations of the CKP hierarchy: overview and character identities}
\author{Iana I. Anguelova}

\address{Department of Mathematics,  College of Charleston,
Charleston SC 29424 }
\email{anguelovai@cofc.edu}

\date{\today}


\begin{abstract}
We discuss the  Hirota bilinear equation for the CKP hierarchy  introduced in \cite{DJKM6},   and  its algebraic properties. We  review in parallel  the two  bosonizations of the CKP hierarchy:  one arising from a twisted  Heisenberg algebra (\cite{OrlovLeur}), and the second from an untwisted  Heisenberg algebra (\cite{AngCKPSecond}). In particular, we  recount  the decompositions into irreducible Heisenberg modules  and the (twisted) fermionic structures of the spaces spanned by the highest weight vectors under the two Heisenberg actions. We show that the  two bosonizations give rise to four different diagonalizable grading operators on the CKP Fock space, not all of them commuting among each other. We compute the various graded dimensions related to these four grading operators. We prove a sum-vs-product  identity relating the bosonic vs fermionic descriptions   under the untwisted Heisenberg action, utilizing the charge and the degree grading operators. As a corollary,  the resulting identities  relate the CKP hierarchy,   the Dyson  crank of a partition and  the Hammond-Lewis birank of a distinct integer bipartition.
\end{abstract}

\maketitle


\section{Introduction}
\label{sec:intro}

There are many ties between the area of integrable hierarchies on one side, and vertex algebras and  conformal field theory on the other. The fruitful relation between those two areas started with the realization that the Hirota bilinear equation associated to the  Kadomtsev-Petviashvili (KP) hierarchy can be concisely written in terms of vertex operators as follows (\cite{DJKM-KP}, \cite{KacRaina}, \cite{Kac}):
\[
Res_z \Big(\psi^+ (z)\otimes \psi^- (z)\Big) (\tau \ten \tau) =0,
\]
where $\psi^+ (z)$ and  $\psi ^-(z)$ are  two fermionic fields
with only nontrivial Operator Product Expansion (OPE)
\[
\psi^+ (z)\psi^- (w)\sim \frac{1}{z-w}\sim \psi^-(z)\psi^+(w).
\]
The vector $\tau$ on which these vertex operators  $\psi^+ (z)$ and  $\psi ^-(z)$ act belongs to (a subspace of) the Fock space which the  Clifford algebra  modes of these vertex operator create from the vacuum vector $|0\rangle$. The vector $\tau$ is often referred to as "the $\tau$ function".

As is well known, the advantage of this purely algebraic form of the Hirota equation (instead of the original Hirota equation, \cite{Hirota}), is twofold: first, it describes the entire hierarchy, not just the KP equation itself. And second,  one can easily prove that this KP  Hirota equation commutes with the action of the   $a_{\infty}$   Lie algebra (\cite{DJKM-KP}, \cite{KacRaina}, \cite{Miwa-book}), and as a result the solutions, i.e., the $\tau$ functions of the KP hierarchy, belong to the  orbit
\[
GL_{\infty}|0\rangle.
\]

Bosonization is the process of writing the vertex operators (in the KP case  $\psi^+ (z)$ and  $\psi ^-(z)$) in terms of exponentiated boson fields. It is a necessary process if one is to write the purely algebraic form of the Hirota equation into a hierarchy of actual differential equations, including  translating  the Fock space element $\tau$ into a function in the differential equation sense, over an infinite set of variables $x_1, x_2, \dots, x_n, \dots $.
As is well known, the bosonization process in the case of the KP hierarchy resulted in the
 boson-fermion correspondence, a vertex algebra isomorphism between the charged free fermions super vertex algebra and the  lattice super vertex algebra of the rank one odd lattice (see e.g. \cite{Kac}).

 In continuation of their work on the KP hierarchy,   Date, Jimbo, Kashiwara and Miwa  introduced two  new hierarchies related to the KP hierarchy: the BKP and the CKP hierarchies (\cite{DJKM-4} and \cite{DJKM6}).  These were initially defined  as  reductions from the KP hierarchy, by assuming conditions on the pseudo-differential operator $L$ in the Lax form used to define those hierarchies.
  For both of them Date, Jimbo, Kashiwara and Miwa suggested a Hirota bilinear equation, i.e., operator approach.
     The bosonization of the BKP case proceeded  similarly to the KP case, and from the point of view of vertex algebras   the bosonization of the BKP hierarchy resulted in the boson-fermion correspondence of type B (\cite{DJKM-4}, \cite{YouBKP}), which was later interpreted as an isomorphism of certain twisted vertex (chiral)  algebras (\cite{Ang-Varna2}, \cite{AngTVA}).

For the CKP hierarchy,   Date, Jimbo, Kashiwara and Miwa suggested in  \cite{DJKM6}  the following Hirota equation:
\begin{equation}\label{eqn:Hirotadef}
Res_z \big(\chi (z)\otimes \chi (-z)\big) (\tau \ten \tau) =0,
\end{equation}
where the field $\chi (z)$ has Operator Product Expansion (OPE)
\[
\chi(z)\chi(w)\sim \frac{1}{z+w}.
\]
If one is to rewrite this purely algebraic Hirota equation into a hierarchy of differential equations, once again one needs the process of bosonization. There are 3 stages to any bosonization:
\begin{enumerate}
\item Construct a bosonic Heisenberg current from the generating fields, hence obtaining a field representation of the Heisenberg algebra on the Fock space;
\item Decompose the Fock space into irreducible Heisenberg modules and determine the structure of the space spanned by the highest weight vectors under the Heisenberg action;
\item Use this decomposition to express the original generating fields in terms of exponential boson fields.
\end{enumerate}
The CKP hierarchy offered its first surprise among many when it was discovered that not one, but two different  bosonizations are possible. The first one was suggested by the original authors of the CKP hierarchy,  Date, Jimbo, Kashiwara and Miwa ( \cite{DJKM6}):  they suggested a twisted Heisenberg field, but did not complete the bosonization following from it.
In \cite{OrlovLeur} van de Leur, Orlov and Shiota  completed this twisted  bosonization and derived some applications of it. The second bosonization was made possible when the CKP hierarchy was related to the $\beta-\gamma$ system of conformal field theory (also called the boson ghost system), see \cite{AngMLB}. As a result of this relation, a second, untwisted  Heisenberg current  was constructed from the generating field $\chi (z)$. This opened the process of the second bosonization  of the CKP hierarchy, which was completed in \cite{AngCKPSecond}.

In this conference proceedings, we start by  discussing two important, albeit surprising, properties of the algebraic Hirota equation \eqref{eqn:Hirotadef} and its Fock space: that there are no finite sum solutions (Lemma \ref{lem:nosolutions}), but that there are series solutions in the series completion of the Fock space  (Proposition \ref{prop:symmetry} and its corollary). Specifically, as the name CKP suggests, Proposition \ref{prop:symmetry} shows that the Hirota equation \eqref{eqn:Hirotadef} commutes with the action of  the $c_{\infty}$ Lie algebra  (this was remarked upon and  used in \cite{OrlovLeur},  but without proof). Next, we
review the two bosonizations of the CKP hierarchy in parallel, by  addressing each of the three stages above. In particular, we discuss the Heisenberg decompositions, as well as the structure of  the  corresponding vector spaces spanned by the highest weight vectors for the two Heisenberg actions. In the untwisted case,  the vector space spanned by the highest weight vectors has a structure equivalent to  the symplectic fermion vertex algebra, as we showed in \cite{AngCKPSecond}. In the twisted case,  although in \cite{OrlovLeur} the authors don't use the language of vertex algebras, their calculations   show that the vector space spanned by the highest weight vectors has a structure of a twisted fermion vertex algebra (in the sense of \cite{ACJ}), see Theorem \ref{thm:symplVA} and Corollary \ref{cor:FDecomp}. In Section II  we also recount two of the gradings we will use in Section III: the charge grading, by the 0-mode  $h^{\mathbb{Z}}_0$ of the untwisted Heisenberg field, \eqref{eq:Heismodes}, and the degree grading, by the 0-mode $L_0$ of one of the Virasoro fields, \eqref{eqn:VirGrading}.

In Section III we introduce two new grading operators $L^t_0$ and $L^h_0$, derived from two of the other Virasoro fields, $L^{\chi} (z)$ (\eqref{eqn:Vir3})  and $L^h (z^2)$  (\eqref{eqn:Vir1}). We show that $L^t_0$ doesn't commute with the grading operators  $h^{\mathbb{Z}}_0$ and $L^h_0$ (Lemma \ref{lem:nonsiimultdiagon}).  $L^t_0$  commutes with $L_0$ though, and we derive the graded dimension with respect to these two gradings, see Proposition \ref{prop:q-t-trace}. The other three grading operators ($h^{\mathbb{Z}}_0$, $L_0$ and $L^h_0$) do commute among each other, and thus a three-parameter graded dimension can be formed. In order to calculate this three-parameter graded dimension we use the concept of  Hammond-Lewis birank of a bipartition (\cite{HL}, \cite{Garvan}). Based on  Corollary \ref{cor:FDecomp},  we show that there is a one-to-one correspondence between the set $\mathfrak{P}_{tdo}$  of distinct partitions with a triangular part describing the Heisenberg decomposition in the untwisted case, Theorem \ref{thm:HeisDecomp}, and the set of distinct integer bipartitions  $\mathfrak{BP}_{DI}$, see Proposition \ref{prop:counting-formula}. The Hammond-Lewis birank is in fact the charge of the highest weight vector uniquely assigned to such distinct integer bipartition. This one-to-one correspondence is not an isometry (it does not preserve the weights), but we derive a formula for  the degree of the highest weight vector assigned to  a given  distinct integer bipartition, which in turn relates the two weights, see \eqref{eqn:weightformula}. This formula allows us to calculate the three-parameter graded dimension, Theorem \ref{thm:triplecharacter}.

An important consequence of any bosonization is that  by calculating the graded dimensions on both the fermionic and the bosonic sides of the correspondence one can obtain identities relating certain product formulas to certain sum formulas. Such a sum-vs-product  identity perfectly illustrates the equality between the fermionic side (the product formulas) and the bosonic side (the sum formulas). In  the classical boson-fermion correspondence (of type A) the sum-vs-product identity relating the two sides is  the Jacobi triple product identity, as was proved  in \cite{Kac}. The sum-vs-product identity for the boson-fermion correspondence of type D-A is also the Jacobi truple product identity, in a slightly different form, as shown in \cite{AngD-A}.  Here, as is typical for the CKP quirks, this identity  is complicated by the fact that the degree operator $L_0$,  which is the most natural grading operator   to use for the CKP hierarchy,  doesn't act as uniformly on the symplectic fermion side with which the highest weight vectors space $\mathit{F^{hwv}_{\chi}}$ identifies. Also, in the CKP case, the "fermionic side" is not purely fermionic, but instead one can view it as a fermion times a boson (see Corollary \ref{cor:FDecomp})
Nevertheless, we  derive the relevant, if more complicated,  sum-vs-product formula for the CKP correspondence in Theorem \ref{thm:sum-vs-product}. The identity that follows  is somewhat surprising, see comment after Theorem  \ref{thm:sum-vs-product}, but we show in the Appendix that it can also be obtained  as a specialization of the Ramanujan Psi summation formula (see e.g. \cite{hardy1940ramanujan}, \cite{AndrewsRam-1}, \cite{AndrewsRam}).
The  graded dimension of the highest weight vectors space $\mathit{F^{hwv}_{\chi}}$ connects the CKP hierarchy  to  the  Dyson's crank of a partition (\cite{Garvan-crank}), via its generating function.   A certain sum of the number of partitions with a given Dyson crank allows us  to answer as to  the number of Heisenberg highest weight vectors with given degree and charge.
Finally, Corollary \ref{cor:dysoncrankcharacter} gives a quadruple identity relating in one of the equalities the Dyson crank with the Hammond-Lewis birank.

\section{The  CKP hierarchy and its two bosonizations: overview}

In \cite{DJKM6}    Date, Jimbo, Kashiwara and Miwa suggested the following Hirota equation associated with  the CKP hierarchy:
\begin{equation}\label{eqn:HirotaC}
Res_z \big(\chi (z)\otimes \chi (-z)\big) (\tau \ten \tau) =0,
\end{equation}
where  the twisted neutral boson field $\chi (z)$, indexed as
\begin{equation}
\chi (z) = \sum _{n\in \mathbb{Z}+1/2} \chi _n z^{-n-1/2},
\end{equation}
has OPE
\begin{equation}
\label{equation:OPE-C}
\chi(z)\chi(w)\sim \frac{1}{z+w}.
\end{equation}
 In terms of commutation relations for the modes $\chi_n, \ n\in \mathbb{Z}+1/2$,  this OPE is equivalent to
\begin{equation}
\label{eqn:Com-C}
[\chi_m, \chi_n]=(-1)^{m-\frac{1}{2}}\delta _{m, -n}1.
\end{equation}
The modes  $\chi_n, \ n\in \mathbb{Z}+1/2$  form a Lie algebra which we denote by $L_{\chi}$. The tau function $\tau$ in the above equation \eqref{eqn:HirotaC} is an element of (a completion of)   the Fock space $\mathit{F_{\chi}}$:
\begin{defn} Denote by  $\mathit{F_{\chi}}$  the induced vacuum module of $L_{\chi}$ with  vacuum vector $|0\rangle $, such that for $n>0$ we have $\chi_n|0\rangle=0$.
The vector space $\mathit{F_{\chi}}$ has a basis
\begin{equation}\label{eqn:monomials}
\{|0\rangle, \ \left(\chi _{-j_k}\right)^{m_k}\dots \left(\chi _{-j_2}\right)^{m_2}\left(\chi _{-j_1}\right)^{m_1}|0\rangle \  \arrowvert \ \ j_k>\dots > j_2 > j_1 > 0, \ j_i\in \mathbb{Z}+\frac{1}{2}, \ m_i > 0, m_i\in \mathbb{Z}, \ i=1, 2, \dots, k\}.
\end{equation}
\end{defn}
\begin{defn}
Define the Hirota operator by
\begin{equation}
S^C= Res_z \ \chi (z)\otimes \chi (-z)
\end{equation}
In modes $S^C$ translates to
\begin{equation}\label{eqn:Hirotamodes}
S^C = \sum _{n\in \mathbb{Z}+1/2} (-1)^{n-\frac{1}{2}}\chi _n \ten \chi _{-n} =  \chi _{\frac{1}{2}} \ten \chi _{-{\frac{1}{2}}} - \chi _{-\frac{1}{2}} \ten \chi _{{\frac{1}{2}}} -  \chi _{\frac{3}{2}} \ten \chi _{-{\frac{3}{2}}} + \chi _{-\frac{3}{2}} \ten \chi _{{\frac{3}{2}}}+\dots .
\end{equation}
\end{defn}
One of the  first  surprises of the CKP hierarchy and its Fock space is  the fact that there are no actual elements of $\mathit{F_{\chi}}$, besides the vacuum vector $|0\rangle$, that solve the algebraic Hirota equation:
\begin{lem}\label{lem:nosolutions} If $v\in \mathit{F_{\chi}}$ solves the Hirota equation \eqref{eqn:HirotaC},  $
Res_z \big(\chi (z)\otimes \chi (-z)\big) (v \ten v) =0$, \
then $v=|0\rangle$.
\end{lem}
\begin{proof}
Let $v\in \mathit{F_{\chi}}$, and let $v\neq |0\rangle$. Then $v$ is a sum of monomials of the form \eqref{eqn:monomials}. Let  $N>0, \ N\in \mathbb{Z}+1/2$ be the largest $N$ such that there is a monomial  among the summands of $v$  which contains $\chi _{-N}$ (i.e., $-N$ is the lowest index among all the indexes of $\chi_j$ present in $v$). Then we can write $v$ in the form
\[
v=\chi_{-N}^mP_m\big( \chi _{-j_k}, \dots \chi _{-j_2},  \chi _{-j_1}\big)|0\rangle +\dots \chi_{-N}P_1\big( \chi _{-j_k}, \dots \chi _{-j_2},  \chi _{-j_1}\big)|0\rangle + P_0\big( \chi _{-j_k}, \dots \chi _{-j_2},  \chi _{-j_1}\big)|0\rangle,
\]
where $P_m, \dots , P_1, P_0$ are polynomials in the variables $\chi _{-j_k}, \dots  \chi _{-j_2}, \chi _{-j_1}$ with  $j_k, \dots j_2, j_1$ all strictly lower than $N$, $P_m\neq 0$, and $m\geq 1$.
Then we have
\begin{align*}
S^C& (v\ten v)= \sum _{n\in \mathbb{Z}+1/2} (-1)^{n-\frac{1}{2}}\chi _n \ten \chi _{-n} v\ten v \\
 &= (-1)^{N-\frac{1}{2}} \big(\chi_N \ten \chi_{-N} -  \chi_{-N}\ten \chi_N \big) v\ten v + \left(\sum _{n=-N+1}^{N-1} (-1)^{n-\frac{1}{2}}\chi _n \ten \chi _{-n}\right) v\ten v \\
&= (-1)^{N-\frac{1}{2}} \Big(\chi_{-N}^{m-1}P_m\big( \chi _{-j_k}, \dots ,   \chi _{-j_1}\big)|0\rangle +\dots P_1|0\rangle\Big)\ten  \Big(\chi_{-N}^{m+1}P_m\big( \chi _{-j_k}, \dots ,  \chi _{-j_1}\big)|0\rangle +\dots +\chi_{-N}P_0|0\rangle\Big) \\
&\ \ \ - (-1)^{N-\frac{1}{2}} \Big(\chi_{-N}^{m+1}P_m\big( \chi _{-j_k}, \dots ,  \chi _{-j_1}\big)|0\rangle +\dots + \chi_{-N}P_0|0\rangle\Big) \ten \Big(\chi_{-N}^{m-1}P_m\big( \chi _{-j_k}, \dots ,  \chi _{-j_1}\big)|0\rangle +\dots P_1|0\rangle\Big)\\
& \ \ \ \ \ +\left(\sum _{n=-N+1}^{N-1} (-1)^{n-\frac{1}{2}}\chi _n \ten \chi _{-n}\right) v\ten v.
\end{align*}
Since all the $j_k, \dots , j_2, j_1$ are strictly lower than $N$, the sum $ \left(\sum _{n=-N+1}^{N-1} (-1)^{n-\frac{1}{2}}\chi _n \ten \chi _{-n}\right) v\ten v$ will have terms that contain at most $\chi_{-N}^{m}$ in each of the tensor products. Thus there is no other similar term to cancel the nonzero $\chi_{-N}^{m+1}P_m\big( \chi _{-j_k}, \dots \chi _{-j_2},  \chi _{-j_1}\big)|0\rangle\ten \chi_{-N}^{m-1}P_m\big( \chi _{-j_k}, \dots \chi _{-j_2},  \chi _{-j_1}\big)|0\rangle $. Hence the sum of tensor products above is clearly nonzero.
\end{proof}
   This shows that there are no finite-sum solutions, in contrast to the KP case where every monomial in the charged free fermion Fock space is actually a solution to the corresponding KP Hirota equation.  Thus one has to necessarily go to a completion $\widetilde{\mathit{F_{\chi}}}$ of  $\mathit{F_{\chi}}$, where one considers series  of monomials instead of finite sums. Since  $\mathit{F_{\chi}}$  can be viewed as isomorphic to  a polynomial algebra, $\mathit{F_{\chi}}$ is dense in such a completion $\widetilde{\mathit{F_{\chi}}}$. On the other hand, in $\widetilde{\mathit{F_{\chi}}}$ we expect to have many solutions, due to the following representation theory considerations.
 As is known from e.g. \cite{WangKac}, \cite{WangDuality}, \cite{ACJ}, the field  $\chi (z)$ is  related to the double-infinite rank Lie algebra $c_{\infty}$. Recall (see \cite{Kac-Lie}), the Lie algebra  $c_{\infty}$ is defined via the  Lie algebra  $\bar{a}_{\infty}$  of infinite matrices of the form
\begin{equation}
\bar{a}_{\infty}=\{(a_{ij}) | \ i, j\in \mathbb{Z}, \ a_{ij} =0 \ \text{for} |i-j|\gg 0 \}.
\end{equation}
As usual denote the elementary matrices by $E_{ij}$.

The algebra $\bar{c}_{\infty}$ is the subalgebra of $\bar{a}_{\infty}$ consisting of the infinite matrices preserving the bilinear form $(v_i; v_j)=(-1)^i \delta_{i, 1-j}$, i.e.,
\begin{equation}
\bar{c}_{\infty}=\{(a_{ij})\in \bar{a}_{\infty} | \ a_{ij}=(-1)^{i+j-1}a_{1-j, 1-i} \}.
\end{equation}
The algebra $c_{\infty}$ is the central extension of
$\bar{c}_{\infty}$ by a central element $c$, $c_{\infty}=\bar{c}_{\infty}\oplus \mathbb{C} c$, with cocycle $C$
given by
\begin{align*}
C(E_{ij},E_{ji})&=-C(E_{ji},E_{ij})=1,\quad \text{if} \enspace i\leq 0,\enspace j\geq 1 \\
C(E_{ij},E_{kl})&=0\quad \text{ in all other cases}.
\end{align*}
The  commutation relations for the elementary matrices in the centrally extended algebra  are
\begin{align*}
[E_{ij},E_{kl}]=\delta_{jk}E_{il}-\delta_{li}E_{kj}+C(E_{ij},E_{kl})c.
\end{align*}
The elementary matrices though are not in $c_\infty$, instead the generators for the algebra $c_\infty$ can be written in terms of these elementary matrices as:
\[
\{ (-1)^j E_{i, j} -(-1)^i E_{1-j, 1-i}, \  i, j \in \mathbb{Z}; \ \text{and} \ \ c\}.
\]
We can arrange the non-central generators in a generating series
\begin{equation}
E^C(z, w) =\sum _{i, j\in \mathbb{Z}} ((-1)^jE_{ij}-(-1)^iE_{1-j, 1-i})z^{i-1}w^{-j}.
\end{equation}
The generating series  $E^C(z,w)$ obeys the  relations: $E^C(z,w) = E^C(w,z)$ (and thus one can view it as  bosonic), and
\begin{align*}
[E^C(z_1, w_1), &E^C(z_2, w_2)]=  E^C(z_1, w_2)\delta(z_2 + w_1) -E^C(z_2, w_1)\delta(z_1+ w_2) \\
&\hspace{1.7cm} - E^C(w_2,w_1)\delta(z_1 + z_2) + E^C(z_1,z_2)\delta(w_2 + w_1)   \\
&+ 2 \iota_{z_1,w_2} \frac{1}{z_1+ w_2} \iota_{w_1,z_2} \frac{1}{w_1  + z_2}c
- 2 \iota_{ w_2, z_1}\frac{1}{ w_2 + z_1} \iota_{z_2,w_1} \frac{1}{z_2 + w_1 }c \\
& + 2 \iota_{z_1,z_2 } \frac{1}{   z_1+ z_2} \iota_{w_1,w_2 } \frac{1}{w_1 + w_2 }c
-  2 \iota_{ z_2,z_1} \frac{1}{   z_2+ z_1} \iota_{ w_2,w_1} \frac{1}{w_2 + w_1 }c.
\end{align*}
Hence we can show that
\begin{prop}
The assignment $E(z, w)\to -:\chi(z)\chi (w):$, \ $c\to -\frac{1}{2}Id_{\mathit{F_C}}$ gives a representation of the Lie algebra $c_{\infty}$ on the Fock space $\mathit{F_{\chi}}$
\end{prop}
Note that this is an important  correction from \cite{ACJ}, Proposition 6.2. In particular, the central charge is negative, which is very important, as it implies that the module $\mathit{F_{\chi}}$ is not an integrable module for the action of $c_{\infty}$. This in itself  implies that  we will need to consider a series completion of $\mathit{F_{\chi}}$ for the purposes of exponentiating the action of $c_{\infty}$ to the corresponding group; which  in light of Lemma \ref{lem:nosolutions} is entirely not surprising.
\begin{prop}\label{prop:symmetry}
The Hirota operator $S^C= Res_z \ \chi (z)\otimes \chi (-z)$ commutes with the action of $c_{\infty}$ on the  space $\mathit{F_{\chi}}\otimes \mathit{F_{\chi}}$, i.e.,
 \[
 \left[ E^C(z_1, w_1)\ten 1 +1\ten E^C(z_1, w_1),\  S^C\right] =0,
 \]
where for ease of notation we identify  $E(z_1, w_1)$ with its representation $-:\chi(z_1)\chi (w_1):$ on $\mathit{F_{\chi}}$. In addition,
\[
S^C \left(|0\rangle \otimes |0\rangle\right) =0.
\]
\end{prop}
A statement similar to this proposition is mentioned in \cite{OrlovLeur}, but without proof.
\begin{proof}
The second property  $S^C \left(|0\rangle \otimes |0\rangle\right) =0$  is trivially satisfied, as the modes representation  \eqref{eqn:Hirotamodes} of the Hirota operator  shows, since  for
 $n>0$ we have  $\chi _n  |0\rangle =0$, and for $n<0$ we have  $\chi _{-n}  |0\rangle =0$.

To prove the first property, we will use  the formal delta function notation (the formal delta-function at $z= w$, for $\lambda =1$, is defined by e.g. \cite{Kac}, it is defined for general $\lambda$ in  \cite{ACJ}):
$$
\delta(z; \lambda w): = \sum_{n\in\mathbb Z}\lambda ^{-n-1} z^nw^{-n-1}.
$$
By abuse of notation one sometimes  writes $\delta(z-\lambda w)$,  even though $\delta(z; \lambda w)$  depends on two formal variables $z$ and $w$, and a parameter $\lambda$. This is especially confusing for $\lambda =-1$, as then we have the rather peculiar fact that  $\delta(z_1+z_2)=\delta(z_1;  -z_2)=-\delta(z_2+z_1)=-\delta(z_2;  -z_1)$.

We calculate by Wick's Theorem
\begin{align*}
[ E^C(z_1, w_1)&\ten 1 +1\ten E^C(z_1, w_1),\  S^C]  \\
 &= Res_z \Big(\left[-:\chi(z_1)\chi (w_1):\ , \  \chi (z)\right] \otimes \chi (-z)\Big) + Res_z \Big(\chi (z) \otimes  \left[-:\chi(z_1)\chi (w_1): \ , \  \chi (-z)\right]\Big)\\
& = -Res_z \Big( \delta (z_1;  - z)\chi (w_1) \otimes \chi (-z)\Big) -Res_z \Big( \delta (w_1; - z)\chi (z_1) \otimes \chi (-z)\Big) \\
& -Res_z \Big( \chi (z)\otimes \delta (z_1;  z)\chi (w_1)\Big) -Res_z \Big( \chi (z)\otimes \delta (w_1;   z)\chi (z_1)\Big)\\
& = -Res_z \Big( \chi (w_1) \otimes \delta (z_1;  - z)\chi (-z)\Big) -Res_z \Big( \chi (z_1) \otimes \delta (w_1; - z)\chi (-z)\Big) \\
& -Res_z \Big( \delta (z_1;   z)\chi (z)\otimes \chi (w_1)\Big) -Res_z \Big( \delta (w_1;  z)\chi (z)\otimes \chi (z_1)\Big)
\end{align*}
Now we use the properties of the delta functions established in \cite{ACJ}, in particular that
\[
\text{Res}_z f(z) \delta(z; \lambda w) =f(\lambda w),
\]
as well as the above-mentioned fact that $\delta(z_1;  -z_2)=-\delta(z_2;  -z_1)$, but  $\delta(z_1;  z_2)= \delta(z_2;   z_1)$,  and we get
\[
[ E^C(z_1 , w_1)\ten 1 +1\ten E^C(z_1, w_1),\  S^C] =  \chi (w_1) \otimes \chi (z_1) +\chi (z_1) \otimes \chi (w_1) - \chi (z_1)\otimes \chi (w_1) - \chi (w_1)\otimes \chi (z_1) =0.
\]
\end{proof}
This proposition ensures that when we exponentiate the action of  $c_{\infty}$ to a  corresponding group  action on the series completion  $\widetilde{\mathit{F_{\chi}}}$ of  $\mathit{F_{\chi}}$, each element of the resulting group orbit of the vacuum vector $ |0\rangle$  will be a solution  to the Hirota equation \eqref{eqn:HirotaC}:
\begin{cor}
Let $\tau =\exp{\left(-\sum c_{m n}:\chi _{m}\chi _{n}:\right)}|0\rangle \in \widetilde{\mathit{F_{\chi}}}$. Then $\tau$ is a solution of the Hirota equation \eqref{eqn:HirotaC}, i.e.,
\[
Res_z \big(\chi (z)\otimes \chi (-z)\big) (\tau \ten \tau) =0.
\]
\end{cor}
Thus the Hirota equation \eqref{eqn:HirotaC} is guaranteed to have solutions in $\widetilde{\mathit{F_{\chi}}}$.
In order to write such solutions as $\tau$ functions solving a hierarchy of actual differential equations, we need to bosonize  the purely algebraic Hirota equation.  Any bosonization has three natural stages:
\begin{enumerate}
\item Construct a bosonic Heisenberg current from the generating fields, hence obtaining a field representation of the Heisenberg algebra on the Fock space;
\item Decompose the Fock space into irreducible Heisenberg modules and determine the structure of the space spanned by the highest weight vectors under the Heisenberg action;
\item Use this decomposition to express the original generating fields in terms of exponential boson fields.
\end{enumerate}
The surprise that not one, but two different bosonizations exist in the CKP case,  stems from the fact that there is not one, but two different Heisenberg fields generated by the field $\chi (z)$ and its descendant field $\chi (-z)$. The first, twisted, Heisenberg representation was suggested in the original paper \cite{DJKM6} introducing the CKP hierarchy, and the bosonization initiated by the twisted Heisenberg current was  studied in \cite{OrlovLeur}.  For the purposes of exposition, in this conference proceedings we will  summarize their results regarding the twisted bosonization, as well as our own. The existence of a second, untwisted Heisenberg field was established in \cite{AngMLB}, and the bosonization initiated by the untwisted Heisenberg current was completed in \cite{AngCKPSecond}. The following proposition summarizes the construction of the two Heisenberg fields from the generating field $\chi (z)$:
\begin{prop}\label{prop:Heis-chi} (\cite{AngMLB})
I. Let
\[
h_\chi^{\mathbb{Z}+1/2}(z)= \frac{1}{2}:\chi (z)\chi(-z):.
\]
We have $h_\chi^{\mathbb{Z}+1/2}(-z)=h_\chi^{\mathbb{Z}+1/2}(z)$, and we index $h_\chi^{\mathbb{Z}+1/2} (z)$ as
$h_\chi^{\mathbb{Z}+1/2} (z)=\sum _{n\in \mathbb{Z}+1/2} h^t_{n} z^{-2n-1}$. The field $h_\chi^{\mathbb{Z}+1/2} (z)$ has OPE with itself given by:
\begin{equation}
\label{eqn:HeisOPEsC-t}
 h^{\mathbb{Z}+1/2}_\chi(z)h^{\mathbb{Z}+1/2}_\chi (w)\sim -\frac{z^2 +w^2}{2(z^2 -w^2)^2}\sim -\frac{1}{4}\frac{1}{(z-w)^2} - \frac{1}{4}\frac{1}{(z+w)^2} ,
\end{equation}
and its  modes, $h^t_n, \ n\in \mathbb{Z}+1/2$, generate a \textbf{twisted} Heisenberg algebra $\mathcal{H}_{\mathbb{Z}+1/2}$ with relations \\ $[h^t_m, h^t_n]=-m\delta _{m+n,0}1$, \ $m,n\in \mathbb{Z}+1/2$.\\
II. Let
\[
h_\chi^{\mathbb{Z}}(z)= \frac{1}{4z}\left(:\chi (z)\chi(z):- :\chi (-z)\chi(-z):\right).
\]
 We have $h_\chi^{\mathbb{Z}}(-z)=h_\chi^{\mathbb{Z}}(z)$, and we index $h_\chi^{\mathbb{Z}} (z)$ as
$h_\chi^{\mathbb{Z}} (z)=\sum _{n\in \mathbb{Z}} h^{\mathbb{Z}}_{n} z^{-2n-2}$. The field $h_\chi^{\mathbb{Z}} (z)$ has OPE with itself given by:
\begin{equation}
\label{eqn:HeisOPEsC-ut}
 h_\chi^{\mathbb{Z}}(z)h_\chi^{\mathbb{Z}} (w)\sim -\frac{1}{(z^2 -w^2)^2},
\end{equation}
and its  modes, $h^{\mathbb{Z}}_n, \ n\in \mathbb{Z}$, generate an \textbf{untwisted} Heisenberg algebra $\mathcal{H}_{\mathbb{Z}}$ with relations $[h^{\mathbb{Z}}_m, h^{\mathbb{Z}}_n]=-m\delta _{m+n,0}1$, \ $m,n\in \mathbb{Z}$.
 \end{prop}
Next we follow with the decomposition in terms of irreducible Heisenberg modules.
The reason a bosonization procedure of the algebraic Hirota equation is essential in order to produce actual differential equations lies in the well known fact   (see e.g. \cite{KacRaina}, \cite{FLM}) that any irreducible highest weight module of the untwisted Heisenberg
algebra $\mathcal{H}_{\mathbb{Z}}$  is isomorphic to the polynomial algebra with infinitely many variables $\mathit{B_\lambda}\cong \mathbb{C}[x_1,
x_2, \dots , x_n, \dots ]$ where $v\mapsto 1$ and we choose the scaling:
\begin{equation}
h^{\mathbb{Z}}_n\mapsto i\partial _{x_{n}}, \quad h^{\mathbb{Z}}_{-n} \mapsto
inx_n\cdot, \quad \text{for any} \ \ n\in \mathbb{N}, \quad h^{\mathbb{Z}}_0\mapsto \lambda\cdot , \ \lambda \in \mathbb{C}.
\end{equation}
Similarly for the twisted Heisenberg algebra:
$\mathcal{H}_{\mathbb{Z}+1/2}$ has (up-to isomorphism) only one irreducible highest weight module $B_{1/2}\cong \mathbb{C}[t_1, t_3, \dots , t_{2n-1}, \dots ]$, via (we choose the same scaling as in \cite{OrlovLeur}):
\begin{equation}
h^t_{\frac{2n-1}{2}}\mapsto \partial _{t_{2n-1}}, \quad h^t_{-\frac{2n-1}{2}} \mapsto
-\frac{2n-1}{2}t_{2n-1}\cdot, \quad \text{for any} \ \ n\in \mathbb{N}.
\end{equation}

Before we proceed with the decomposition in terms of irreducible Heisenberg modules we first need to recall the grading operators that act on  $\mathit{F_{\chi}}$. In \cite{AngCKPSecond}  we introduced two important gradings on  $\mathit{F_{\chi}}$, one derived from  the action of $h^{\mathbb{Z}}_0$ of the untwisted Heisenberg field (the charge grading), and the second from one of the families of Virasoro fields that we discussed in \cite{AngMLB} (the degree grading). There are also two additional  grading operators which we will introduce in the next section.

 The \emph{charge}   grading $chg$ is derived from the action of the "charge" operator $h^{\mathbb{Z}}_0$:
 \begin{equation}
\label{eq:Heismodes}
h^{\mathbb{Z}}_0=\sum _{k\in \mathbb{Z}_{\geq 0}+1/2}:\chi_{-k} \chi_{k}: =\chi_{-\frac{1}{2}}\chi_{\frac{1}{2}} +\chi_{-\frac{3}{2}}\chi_{\frac{3}{2}} +\dots,
\end{equation}
where $:\ :$ denotes the normal ordered product defined in the usual way (see e.g. \cite{AngCKPSecond}).

$h^{\mathbb{Z}}_0$ is diagonalizable  on $\mathit{F_{\chi}}$ and thus
 it gives $\mathit{F_{\chi}}$ a $\mathbb{Z}$ grading, which we  called \emph{charge} and denote $chg$ (as it is similar to the charge grading in the usual boson-fermion correspondence,  of type A), by
 \begin{equation}
chg \big(|0\rangle \big) =0; \quad
chg  \Big(\left(\chi _{-j_k}\right)^{m_k}\dots \left(\chi _{-j_2}\right)^{m_2}\left(\chi _{-j_1}\right)^{m_1}|0\rangle \Big) =\sum_{j_i\in 2\mathbb{Z} +1/2} m_i - \sum_{j_i\in 2\mathbb{Z} -1/2} m_i.
\end{equation}
Example: $chg \big( \chi_{-\frac{1}{2}} |0\rangle \big) =1$; \   $chg \big( \chi_{-\frac{3}{2}} |0\rangle \big) =-1$;\  $chg \big( \chi_{-\frac{3}{2}} \chi_{-\frac{1}{2}} |0\rangle \big) =0$.

Denote the linear span of monomials  of charge $n$ by $\mathit{F^{(n)}_{\chi}}$. The Fock space $\mathit{F_{\chi}}$ has a charge decomposition
\[
\mathit{F_{\chi}} =\oplus_{n\in \mathbb{Z}} \mathit{F^{(n)}_{\chi}}, \quad \text{with}\quad
h^{\mathbb{Z}}_0 v =chg (v)\cdot v =nv, \quad \text{for \ any} \ v\in \mathit{F^{(n)}_{\chi}}.
\]

The second grading we used in \cite{AngCKPSecond} is the \emph{degree} grading, resulting from the action of one of the Virasoro fields. Namely, we considered the Virasoro field with central charge $c=-\frac{1}{4}$ given by
\begin{equation}
\label{eqn:Vir2}
L^{-\frac{1}{4}} (z^2)= \sum_{n\in \mathbb{Z}}L_n (z^2)^{-n-2} = \sum_{n\in \mathbb{Z}}\Big(\sum_{k+l=n}  \frac{k+1-3l}{4}:\chi_{2k+\frac{1}{2}}\chi_{2l-\frac{1}{2}}:\Big)(z^2)^{-n-2},
\end{equation}
in particular
\begin{equation}
\label{eqn:VirGrading}
L_0 = \frac{1}{2}\Big( \frac{1}{2}:\chi_{-\frac{1}{2}}\chi_{\frac{1}{2}}: -\frac{3}{2}:\chi_{-\frac{3}{2}}\chi_{\frac{3}{2}}: +\frac{5}{2}:\chi_{-\frac{5}{2}}\chi_{\frac{5}{2}}: -\dots \Big).
\end{equation}
Hence
\begin{equation}
L_0 \Big(\left(\chi _{-j_k}\right)^{m_k}\dots \left(\chi _{-j_2}\right)^{m_2}\left(\chi _{-j_1}\right)^{m_1}|0\rangle \Big)= \frac{1}{2}\left(m_k\cdot j_k +\dots m_2\cdot j_2 +m_1\cdot j_1\right)\Big(\left(\chi _{-j_k}\right)^{m_k}\dots \left(\chi _{-j_2}\right)^{m_2}\left(\chi _{-j_1}\right)^{m_1}|0\rangle \Big),
\end{equation}
$j_k>\dots > j_2 >j_1 > 0, \ j_i\in \mathbb{Z}+\frac{1}{2}$,\  $m_i > 0, m_i\in \mathbb{Z}, \ i=1, 2, \dots, k$.
The degree grading is a $\frac{1}{2}\mathbb{Z}$ grading, given by the action of $2L_0$, i.e.,
\[
2L_0 v =deg (v)\cdot v,
\]
 where $v$ is any monomial $\left(\chi _{-j_k}\right)^{m_k}\dots \left(\chi _{-j_2}\right)^{m_2}\left(\chi _{-j_1}\right)^{m_1}|0\rangle$. The degree grading is also used in \cite{OrlovLeur}, although without its connection to the Virasoro field.

Next, we will need some notations for the corresponding indexing sets in the decompositions.
\begin{notation}
Let $\mathcal{ODP}$ denote the set of distinct partitions of the type
\begin{equation}
\mathcal{ODP}=\{\mathfrak{p} =(\lambda_1, \lambda_2, \dots , \lambda_k) \ |\  \lambda_1>\lambda_2>\dots >\lambda_k, \lambda_i \in \frac{1}{2}+\mathbb{Z}_{\geq 0}, \ i=1, \dots , k \}.
\end{equation}
Denote by $T_m$  the $m$-th triangular number---  $T_m:=1+2+\dots +m =\frac{m(m+1)}{2}$, with $T_0 =0$. Let
$\mathfrak{P}_{tdo}$ denote the set  of distinct partitions of the type
\begin{equation}
\mathfrak{P}_{tdo}=\{\mathfrak{p} =(T_m, \lambda_1, \lambda_2, \dots , \lambda_k) \ |\ T_m-\text{triangular\ number}, \ \lambda_1>\lambda_2>\dots >\lambda_k, \lambda_i \in \frac{1}{2}+\mathbb{Z}_{\geq 0}, \ i=1, \dots , k \}.
\end{equation}
As usual, the weight $| \mathfrak{p} | $   of a partition $\mathfrak{p}$ is the sum of its parts, $| \mathfrak{p} |: =T_m +  \lambda_1 + \lambda_2 + \dots + \lambda_k$.
\end{notation}
 We can now formulate the decomposition in terms of Heisenberg modules.
\begin{thm} \label{thm:HeisDecomp}
I. For the action of the twisted Heisenberg algebra $\mathcal{H}_{\mathbb{Z}+1/2}$ on $\mathit{F_{\chi}}$,  the number of highest weight vectors  of degree $n\in \frac{1}{2}\mathbb{Z}$  equals the number of partitions $\mathfrak{p}\in \mathcal{ODP}$ of weight $n$. Thus as twisted Heisenberg modules
\begin{equation}
\mathit{F_{\chi}}\cong  \oplus_{\mathfrak{p}\in \mathcal{ODP}}   \mathbb{C}[t_1, t_3, \dots , t_{2n-1}, \dots ].
\end{equation}
II. (\cite{AngCKPSecond})
For the action of the untwisted Heisenberg algebra $\mathcal{H}_{\mathbb{Z}}$ on $\mathit{F_{\chi}}$,  the number of highest weight vectors  of degree $n\in \frac{1}{2}\mathbb{Z}$  equals the number of partitions $\mathfrak{p}\in \mathfrak{P}_{tdo}$ of weight $n$. Thus as untwisted Heisenberg modules
\begin{equation}
\mathit{F_{\chi}}\cong  \oplus_{\mathfrak{p}\in \mathfrak{P}_{tdo}}   \mathbb{C}[x_1, x_2, \dots , x_n, \dots ].
\end{equation}
\end{thm}
The Heisenberg decomposition in the twisted case (the first part of the theorem above) is not stated in \cite{OrlovLeur}, but can be implied from the calculations there. The second part of this theorem is established in \cite{AngCKPSecond}. The proof of the first part is similar to the proof of the second part, and so we omit it.
\begin{example} The first few highest weight vectors for $\mathcal{H}_{\mathbb{Z}}$ arranged by degree are
\begin{displaymath}
\begin{array}{c|c}
1 & |0\rangle \hspace{\stretch{1}}\\
q^{\frac{1}{2}} & \chi_{-\frac{1}{2}}|0\rangle  \hspace{\stretch{1}}\\
q & \chi_{-\frac{1}{2}}^2|0\rangle  \hspace{\stretch{1}}\\
2q^{\frac{3}{2}} & \chi_{-\frac{1}{2}}^3|0\rangle ;\ \  \chi_{-\frac{3}{2}}|0\rangle  \hspace{\stretch{1}}\\
q^2 & \chi_{-\frac{1}{2}}^4|0\rangle   \hspace{\stretch{1}}\\
2q^{\frac{5}{2}} & \chi_{-\frac{1}{2}}^5|0\rangle ; \ \ \chi_{-\frac{3}{2}} \chi_{-\frac{1}{2}}^2|0\rangle  +2\chi_{-\frac{5}{2}}|0\rangle  \hspace{\stretch{1}}\\
3q^3 & \chi_{-\frac{1}{2}}^6|0\rangle ;\ \ \chi_{-\frac{3}{2}}\chi_{-\frac{1}{2}}^3|0\rangle +3 \chi_{-\frac{5}{2}}\chi_{-\frac{1}{2}}|0\rangle; \ \  \chi_{-\frac{3}{2}}^2|0\rangle  \hspace{\stretch{1}}
 \end{array}
 \end{displaymath}
Two important families of examples  of highest weight vectors for $\mathcal{H}_{\mathbb{Z}}$ are the families  $\chi_{-\frac{1}{2}}^n|0\rangle$ and \ $\chi_{-\frac{3}{2}}^n|0\rangle$. They are the "minimal degree" highest weight vectors: at each negative fixed charge $-n<0$, the vector $\chi _{-3/2}^n |0\rangle$ is the vector with minimal degree with  charge $-n$; and at each fixed charge $n>0$, the vector $\chi _{-3/2}^n |0\rangle$ is the vector with minimal degree of that charge $n$.

There are no similar families  of highest weight vectors for $\mathcal{H}_{\mathbb{Z}+1/2}$, as there is no equivalent charge operator in $\mathcal{H}_{\mathbb{Z}+1/2}$. The first few highest weight vectors by degree are
\begin{displaymath}
\begin{array}{c|c}
1 & |0\rangle \hspace{\stretch{1}}\\
q^{\frac{1}{2}} & \chi_{-\frac{1}{2}}|0\rangle  \hspace{\stretch{1}}\\
q & \text{none}  \hspace{\stretch{1}}\\
q^{\frac{3}{2}} & \chi_{-\frac{3}{2}}|0\rangle -\frac{1}{3}\chi_{-\frac{1}{2}}^3|0\rangle   \hspace{\stretch{1}}\\
q^2 & \chi_{-\frac{3}{2}}\chi_{-\frac{1}{2}}|0\rangle-\frac{1}{6}\chi_{-\frac{1}{2}}^4|0\rangle   \hspace{\stretch{1}}\\
q^{\frac{5}{2}} & \chi_{-\frac{5}{2}}|0\rangle -\chi_{-\frac{3}{2}}\chi_{-\frac{1}{2}}^2|0\rangle +\frac{1}{10}\chi_{-\frac{1}{2}}^5|0\rangle   \hspace{\stretch{1}}
 \end{array}
 \end{displaymath}
\end{example}

Next, we proceed with the structure of the spaces spanned by the corresponding highest weight vectors.
\begin{notation}
Denote by $\mathit{F^{hwv}_\chi}$ the vector space spanned by the highest weight vectors for the untwisted Heisenberg algebra representation on $\mathit{F_{\chi}}$, and by $\mathit{F^{t-hwv}_\chi}$ the vector space spanned by the highest weight vectors for the twisted Heisenberg algebra representation on $\mathit{F_{\chi}}$.
\end{notation}
From the Proposition above, we have the following isomorphisms as vectors spaces:
\begin{align}\label{eqn:isoHWV-Polyn}
\mathit{F_\chi}&\cong \mathit{F^{hwv}_\chi}\ten \mathbb{C}[x_1,
x_2, \dots , x_n, \dots ],\\
\label{eqn:isoHWV-Polyn2}
\mathit{F_\chi}&\cong \mathit{F^{t-hwv}_\chi}\ten \mathbb{C}[t_1,
t_3, \dots , t_{2n-1}, \dots ].
\end{align}
Each of the spaces $\mathit{F^{hwv}_\chi}$ and $\mathit{F^{t-hwv}_\chi}$ has  additional structure, which we need in order to complete the corresponding bosonizations.
\begin{defn}
Let
\begin{equation}
V^-(z) = \exp \Big(-\sum_{n>0}\frac{1}{n}h^{\mathbb{Z}}_n z^{-2n}\Big);\quad \quad V^-(z)^{-1} = \exp \Big(\sum_{n>0}\frac{1}{n} h^{\mathbb{Z}}_n z^{-2n}\Big);
\end{equation}
and
\begin{equation}
V^+(z) = \exp \Big(\sum_{n>0}\frac{1}{n}h^{\mathbb{Z}}_{-n} z^{2n}\Big);\quad \quad
V^+(z)^{-1} = \exp \Big(-\sum_{n>0}\frac{1}{n} h^{\mathbb{Z}}_{-n} z^{2n}\Big)
\end{equation}
Define
\begin{equation}
\beta_\chi (z^2) =\frac{\chi(z) -\chi (-z)}{2z}; \quad \gamma_\chi (z^2) =\frac{\chi(z) +\chi (-z)}{2}.
\end{equation}
and
\begin{equation} \label{eqn:Hbeta-Hgamma}
H^{\beta} (z^2) = V^+(z)^{-1}\beta_\chi(z^2)z^{-2h^{\mathbb{Z}}_0}V^-(z)^{-1}, \quad \quad H^{\gamma} (z^2) = V^+(z)\gamma_\chi (z^2)z^{2h^{\mathbb{Z}}_0}V^-(z).
\end{equation}
( $V^-(z)$ and $V^+(z)$ are actually functions of $z^2$, so the notation is unambiguous).

Similarly denote
\begin{equation}
V_t^-(z) = \exp \Big(\sum_{n>0}\frac{2}{2n-1}h^t_{\frac{2n-1}{2}} z^{-2n+1}\Big);\quad \quad V_t^-(z)^{-1} = \exp \Big(-\sum_{n>0}\frac{2}{2n-1} h^t_{\frac{2n-1}{2}} z^{-2n+1}\Big);
\end{equation}
\begin{equation}
V_t^+(z) = \exp \Big(-\sum_{n>0}\frac{2}{2n-1}h^t_{-\frac{2n-1}{2}} z^{2n-1}\Big);\quad \quad V_t^+(z)^{-1} = \exp \Big(\sum_{n>0}\frac{2}{2n-1} h^t_{-\frac{2n-1}{2}} z^{2n-1}\Big);
\end{equation}
and
\begin{equation}
H^{\chi} (z) =V_t^+(z)^{-1} \chi (z) V_t^-(z)^{-1}.
\end{equation}
\end{defn}
\begin{thm} \label{thm:symplVA}
I.  (\cite{AngCKPSecond}) The vector space  $\mathit{F^{hwv}_\chi}$ spanned by the highest weight vectors for the untwisted Heisenberg algebra  $\mathcal{H}_{\mathbb{Z}}$ has a structure of a super vertex algebra,  strongly generated by the fields $H^{\beta}(z)$ and $H^{\gamma}(z)$, with vacuum vector $|0\rangle$, translation operator $T=L^{hwv}_{-1}$, and vertex operator map induced by \begin{equation}
Y( \chi _{-1/2}|0\rangle, z)=H^{\gamma} (z), \quad Y( \chi _{-3/2}|0\rangle, z)=H^{\beta} (z).
\end{equation}
This vertex algebra structure is a realization of  the symplectic fermion vertex algebra, indicated by the OPEs:
\begin{align}
&H^{\beta} (z) H^{\gamma} (w) \sim \frac{1}{(z -w)^2},  \quad  H^{\gamma} (z)H^{\beta} (w)\sim -\frac{1}{(z -w)^2}; \\
&H^{\beta} (z) H^{\beta} (w)\sim  0; \quad H^{\gamma} (z) H^{\gamma} (w) \sim 0.
\end{align}
II. The vector space  $\mathit{F^{t-hwv}_\chi}$ spanned by the highest weight vectors for the twisted Heisenberg algebra  $\mathcal{H}_{\mathbb{Z}+1/2}$ has a structure of an $N=2$  twisted vertex algebra,   generated by the field $H^{\chi}(z)$,  with vacuum vector $|0\rangle$, and vertex operator map induced by \begin{equation}
Y( \chi _{-1/2}|0\rangle, z)=H^{\chi} (z).
\end{equation}
This twisted vertex algebra structure is twisted  fermionic, indicated by the OPEs:
\begin{equation}\label{eqn:twistedfermionOPE}
H^{\chi} (z) H^{\chi} (w) \sim \frac{z-w}{(z +w)^2}.
\end{equation}
\end{thm}
The first part of this theorem was proved in \cite{AngCKPSecond}. The second part uses the notion of a twisted vertex algebra developed in \cite{AngTVA} and \cite{ACJ}. Its proof is similar to the proof of the first part, and  quite technical, so we omit it. The basic OPE indicated in \eqref{eqn:twistedfermionOPE} can be read from the calculations in \cite{OrlovLeur}.
\begin{cor} \label{cor:FDecomp}  I. Define (\cite{Abe})
\begin{align*}
 \mathit{SF}: = \{H^{\beta}_{(m_k)}\dots H^{\beta}_{(m_2)}&H^{\beta}_{(m_1)}H^{\gamma}_{(n_s)}\dots H^{\gamma}_{(n_2)}H^{\gamma}_{(n_1)}|0\rangle \  \arrowvert \\ & \arrowvert \  m_k<\dots m_2 <m_1, \  n_s<\dots n_2<n_1; \  \  m_i, n_j \in \mathbb{Z}_{<0}, \ i=1, 2, \dots, k; j=1, 2, \dots s\}.
\end{align*}
We have as vertex algebras
\begin{equation}
\mathit{F^{hwv}_\chi}  \cong  \mathit{SF};
\end{equation}
and as vector spaces
\begin{equation}
 \mathit{F_\chi}\cong \mathit{F^{hwv}_\chi}\ten \mathbb{C}[x_1,
x_2, \dots , x_n, \dots ]\cong \mathit{SF} \ten \mathbb{C}[x_1,
x_2, \dots , x_n, \dots ].
\end{equation}
II. Define
\begin{equation}
 \mathit{SF^t}: = \{H^{\chi}_{(n_s)}\dots H^{\chi}_{(n_2)}H^{\chi}_{(n_1)}|0\rangle \   \arrowvert \   n_s<\dots n_2<n_1; \  n_j \in \mathbb{Z}+1/2, n_j<0, \ j=1, 2, \dots s\}.
\end{equation}
We have as twisted vertex algebras
\begin{equation}
\mathit{F^{t-hwv}_\chi}  \cong  \mathit{SF^t};
\end{equation}
and as vector spaces
\begin{equation}
 \mathit{F_\chi}\cong \mathit{F^{t-hwv}_\chi}\ten \mathbb{C}[t_1,
t_3, \dots , t_{2n-1}, \dots ]\cong \mathit{SF^t} \ten \mathbb{C}[t_1,
t_3, \dots , t_{2n-1}, \dots ].
\end{equation}
\end{cor}
The Theorem above allows us to express the generating field in terms of the (correspondingly twisted or untwisted) Heisenberg fields, plus the (twisted or untwisted) symplectic fermion fields, thereby completing the 3rd stage of the bosonization process. From \cite{OrlovLeur}, we   have for the twisted bosonization
\begin{equation} \label{eqn:chi-bos}
\chi(z)=V_t^+(z)H^{\chi} (z)V_t^-(z).
\end{equation}
Unfortunately, we do not know if the field $H^{\chi} (z)$ can be bosonized further.  In \cite{OrlovLeur}  super, or fermionic, variables were introduced to describe the space $\mathit{F^{t-hwv}_\chi}$. I.e., introduce the Grassmann variables $t_{\frac{2n-1}{2}}$, $n\in \mathbb{Z}$, so that
\begin{equation}
t_{\frac{2m-1}{2}}t_{\frac{2n-1}{2}} =-t_{\frac{2n-1}{2}}t_{\frac{2m-1}{2}}, \ \text{for}\ m\neq n; \quad \left(t_{\frac{2n-1}{2}}\right)^2 =0.
\end{equation}
Consequently, we consider the Grassmann algebra generated by the Grassmann variables  $t_{\frac{2n-1}{2}}$, $n\in \mathbb{Z}$, which  we will denote as $\Lambda[t_{\frac{1}2},
t_{\frac{3}2}, \dots , t_{\frac{2n-1}{2}}, \dots ]$. The Grassman variables $t_{\frac{2n-1}{2}}$, $n\in \mathbb{Z}$, commute with the commutative variables $t_{2n-1}$, $n\in \mathbb{Z}$. We of course have $\Lambda[t_{\frac{1}2},
t_{\frac{3}2}, \dots , t_{\frac{2n-1}{2}}, \dots ]\cong \mathit{SF^t}$.
\begin{thm} (\cite{OrlovLeur}
The twisted bosonization of the CKP hierarchy is the  isomorphism $\sigma^t$ between the twisted vertex algebra generated by the field $\chi (z)$ and its descendant $\chi (-z)$ on  the Fock space $\mathit{F_{\chi}}$ and the twisted vertex algebra generated by the field $\sigma^t \chi(z)(\sigma^t)^{-1}:=V_{\sigma t}^+(z)H^{\sigma \chi} (z)V_{\sigma t}^-(z)$, where the invertible map $\sigma^t$ is defined as
\begin{equation}
\sigma^t: \mathit{F_{\chi}} \to \Lambda[t_{\frac{1}2},
t_{\frac{3}2}, \dots , t_{\frac{2n-1}{2}}, \dots ] \ten \mathbb{C}[t_1,
t_3, \dots , t_{2n-1}, \dots ], \quad |0\rangle \mapsto 1
\end{equation}
by
\begin{align}
\sigma^t h^t_{\frac{2n-1}{2}}(\sigma^t)^{-1}: &=\partial_{t_{2n-1}}, \quad \sigma^t h^t_{-\frac{2n-1}{2}}(\sigma^t)^{-1}: =-\frac{2n-1}{2}t_{2n-1}\cdot, \quad \text{for any} \ \ n\in \mathbb{N}; \\
\sigma^t H^{\chi}_{(n)}(\sigma^t)^{-1} &= H^{\sigma \chi}_{(n)}: =\partial _{t_{\frac{2n-1}{2}}}, \quad \sigma^t H^{\chi}_{(-n)}(\sigma^t)^{-1}: =-\frac{2n-1}{2}t_{\frac{2n-1}{2}}\cdot, \quad \text{for any} \ \ n\in \mathbb{N}.
\end{align}
We have
\begin{equation}
V_{\sigma t}^+(z) = \sigma^t V_{ t}^+(z) (\sigma^t)^{-1} : = \exp \Big(\sum_{n>0} t_{2n-1} z^{2n-1}\Big), \quad
V_{\sigma t}^-(z) = \sigma^t V_{ t}^-(z) (\sigma^t)^{-1} : = \exp \Big(\sum_{n>0}\frac{2}{2n-1}\partial_{t_{2n-1}} z^{-2n+1}\Big);
\end{equation}
and
\begin{equation}
H^{\sigma \chi} (z) = \sigma^t H^{\chi} (z) (\sigma^t)^{-1}: =\sum_{n>0} \partial _{t_{\frac{2n-1}{2}}} z^{-n} -\sum_{n>0} \frac{2n-1}{2}t_{\frac{2n-1}{2}}z^n.
\end{equation}
\end{thm}
Now we move to the second, untwisted bosonization.  We have
\[
\chi(z) =\gamma_\chi(z^2) +z \beta_\chi (z^2),
\]
and the fields $\beta_\chi(z^2)$ and $\gamma_\chi (z^2)$
can be written as
\begin{equation} \label{eqn:beta-gamma-bos}
\beta_\chi(z^2)=V^+(z)H^{\beta} (z^2)V^-(z)z^{2h_0}, \quad \quad  \gamma_\chi (z^2) =V^+(z)^{-1}H^{\gamma} (z^2)V^-(z)^{-1}z^{-2h_0}.
\end{equation}
In this case of the untwisted bosonization, even though they are again fermionic, as in the twisted case, the fields $H^{\beta} (z)$ and $H^{\gamma} (z)$ can be bosonized further, via the identification with the symplectic fermions (for more details see \cite{AngCKPSecond}), namely
\begin{align}\label{eqn:ident1}
H^{\beta} (z) &\to e^{-\alpha}_y(z) =\exp (-\sum _{n\ge 1}y_n z^n)\exp (\sum _{n\ge 1}\frac{\partial}{n\partial y_n} z^{-n})e^{-\alpha}z^{-h^y_0}\\ \label{eqn:ident2}
H^{\gamma} (z) &\to \partial_z e^{\alpha}_y(z) =:h^y (z)\exp (\sum _{n\ge 1} y_n z^n)\exp (-\sum _{n\ge 1}\frac{\partial}{n\partial y_n} z^{-n})e^{\alpha}z^{h^y_0}:.
\end{align}
where the lattice fields $e^{\alpha}_y (z)$, $e^{-\alpha}_y (z)$  act on
the bosonic  vector space $\mathbb{C}[e^{\alpha}, e^{-\alpha}]\ten \mathbb{C}[y_1, y_2, \dots , y_n\dots]$  by
\begin{align*}
e^{\alpha}_y(z)& =\exp (\sum _{n\ge 1} y_n z^n)\exp (-\sum _{n\ge 1}\frac{\partial}{n\partial y_n} z^{-n})e^{\alpha}z^{\partial_{\alpha}},\\
e^{-\alpha}_y(z)& =\exp (-\sum _{n\ge 1}y_n z^n)\exp (\sum _{n\ge 1}\frac{\partial}{n\partial y_n} z^{-n})e^{-\alpha}z^{-\partial_{\alpha}},
\end{align*}
We use the index $y$ in $e^{\alpha}_y (z)$, $e^{-\alpha}_y (z)$
 to indicate these are  the exponentiated boson fields acting on the variables $y_1, y_2, \dots , y_n\dots$. We introduce similarly the Heisenberg field $h^y (z)$,
\begin{equation}
h^y (z) = \sum _{n\ge 1}\frac{\partial}{\partial y_n} z^{-n-1} +h^y_0 z^{-1} + \sum _{n\ge 1}n y_n z^{n-1},
\end{equation}
where $h^y_0$ acts on $\mathbb{C}[e^{\alpha}, e^{-\alpha}]\ten \mathbb{C}[y_1, y_2, \dots , y_n\dots]$  by
$h^y_0 e^{m\alpha} P(y_1, y_2, \dots , y_n\dots) = m e^{m\alpha} P(y_1, y_2, \dots , y_n\dots)$.

Finally, the complete second bosonization of the CKP hierarchy is summarized in the following:
\begin{thm}\label{thm:completebosonization}(\cite{AngCKPSecond})
The generating field $\chi (z)$ of the CKP hierarchy can be written as
\[
\chi(z) =\gamma_\chi(z^2) +z \beta_\chi (z^2),
\]
where the fields $\beta_\chi (z)$ and $\gamma_\chi(z)$ can be bosonized as follows:
\begin{align} \label{eqn:finalcompl1}
\sigma: \ \beta_\chi (z) &\mapsto  \ \exp \Big(i\sum_{n>0}(x_n +iy_n) z^{n}\Big)\exp \Big(-i\sum_{n>0}\frac{1}{n}\left(\frac{\partial}{\partial x_n}+i\frac{\partial}{\partial y_n}\right) z^{-n}\Big)e^{-\alpha},\\ \label{eqn:finalcompl2}
\sigma: \ \gamma_\chi(z) &\mapsto \  :\exp \Big(-i\sum_{n>0}(x_n +iy_n)z^{n}\Big) h^y (z) \exp \Big(i\sum_{n>0}\frac{1}{n}\left(\frac{\partial}{\partial x_n}+i\frac{\partial}{\partial y_n}\right) z^{-n}\Big) e^{\alpha}:.
\end{align}
The Fock space $\mathit{F_{\chi}}$ is mapped via the bosonization map $\sigma$ to a subspace of the bosonic space $\mathbb{C}[e^{\alpha}, e^{-\alpha}]\ten \mathbb{C}[x_1, x_2, \dots , x_n, \dots ;y_1, y_2, \dots , y_n\dots]$, with \ $|0\rangle \mapsto 1$.
The Hirota equation \eqref{eqn:HirotaC} is equivalent to
\begin{equation}
Res_z \Big(\beta_\chi (z)\otimes \gamma_\chi (z)- \gamma_\chi (z)\otimes \beta_\chi (z)\Big) =0.
\end{equation}
\end{thm}
 This theorem  allows us to write the purely algebraic Hirota equation as an infinite hierarchy  of actual differential equations. One proceeds similarly to the exposition in \cite{KacRaina}, Chapter 7, by employing the Hirota derivatives technique. These calculations and the types of solutions produced will be detailed in a separate article, as they require some length. In the next section, we will instead focus on one of the applications of these bosonizations, namely  the various characters (graded dimensions) and the identities one can derive from the comparison between the purely bosonic vs the boson+symplectic fermion sides.

\section{Graded dimensions and character identities}

As an application to any bosonization, one can typically obtain identities relating certain product formulas on the fermionic side  to certain sum formulas on the bosonic side.
Recall for instance that  the Jacobi triple product identity can be obtained from such a bosonization --- it was derived  in \cite{Kac} for the classical boson-fermion correspondence of type A, and in \cite{AngD-A} for the bosonization of type D-A. Such a sum-vs-product  identity perfectly illustrates the equality between the fermionic side (the product formulas) and the bosonic side (the sum formulas). Here, as is typical for the CKP quirks, we will show that the sum-vs-product identity representing the untwisted bosonization is much more complicated due to the fact that the degree operator $L_0$ that we had to use for the Heisenberg decomposition doesn't act uniformly on the symplectic fermion side with which the highest weight vectors space $\mathit{F^{hwv}_{\chi}}$ identifies. But first,  we will explore the fact that due to the two bosonizations, in the CKP case we have additional grading operators, besides the charge and the degree gradings.

 The degree grading on the Fock space $\mathit{F_{\chi}}$ is given by the operator $L_0$, but in
fact, we can consider a  one-parameter $\lambda$-degree grading on $\mathit{F_{\chi}}$. Recall, the degree grading is induced by the action of the 0-mode $L_0$ of the Virasoro field $L^{-\frac{1}{4}} (z^2)$. This Virasoro field is an element of the  more general family of Virasoro fields $L^{\lambda} (z)$,  translated   from the $\beta-\gamma$ system (\cite{AngMLB}, for simplicity parameter $\mu$ is set to 0):
\begin{equation}\label{eqn:Vir-lambda}
L^{\lambda} (z)=\lambda :\left(\partial_z\beta (z)\right)\gamma (z): +(\lambda +1):\beta (z)\left(\partial_z\gamma (z)\right): .
\end{equation}
We have
\begin{equation}
\label{eqn:VirGradingLambda}
L^\lambda_0 = -\sum_{n\in \mathbb{Z}}\left(\lambda  +k\right) :\chi_{-2k+\frac{1}{2}}\chi_{2k-\frac{1}{2}}:.
\end{equation}
We want to note that for both the operators $L^\lambda_0$ and  $h^{\mathbb{Z}}_0$ the monomials from \eqref{eqn:monomials}   form an eigenspace basis, as we showed in the previous section.

We made the particular  choice to use $\lambda =-\frac{1}{4}$ for $L_0$ not only because it would  simplify the notation --- for general $\lambda$ the behavior of the operator $L^\lambda_0$ is very similar to that of $L_0$. But also this  is the most natural choice of grading operator for $\mathit{F_{\chi}}$, in particular, as we will see it is the only choice from the family  $L^\lambda_0$ that commutes with the operator $L^t$ which we will introduce below.

The third grading on $\mathit{F_{\chi}}$, corresponding most closely  to the twisted bosonization, is obtained from the Virasoro field  $L^{\chi} (z)$ with central charge $c=1$ (\cite{FMS}, \cite{AngMLB}, for simplicity we set $\kappa =0$):
\begin{equation}
   \label{eqn:Vir3}
   L^{\chi} (z)=\left(-\frac{1}{2z^2}:h^{\mathbb{Z}+1/2} (z)^2: +\frac{1}{16z^4}\right).
   \end{equation}
We have
\begin{equation}
\label{eqn:VirGradingchi}
L^{\chi}_0 = -\frac{1}{2}\sum_{n\in \mathbb{Z}+1/2} :h^t_{-n}h^t_n: +\frac{1}{16}.
\end{equation}
For simplicity we will consider the modified grading operator $L^t_0$, where
\begin{equation}
L^t_0 = L^{\chi}_0  -\frac{1}{16} = -\sum_{\substack{n\in \mathbb{Z}+1/2\\ n>0}} :h^t_{-n}h^t_n:.
\end{equation}
$L^t_0$ is diagonalizable, due to Proposition \ref{thm:HeisDecomp}, but we indeed are always forced  to use the Heisenberg decomposition of Theorem \ref{thm:HeisDecomp} to calculate the action of $L^t_0$ on the elements of $\mathit{F_{\chi}}$ (example  of such calculation is given below, in the proof of Lemma \ref{lem:nonsiimultdiagon}). This is because, unlike for the operators $L^\lambda_0$ and  $h^{\mathbb{Z}}_0$, the monomials from \eqref{eqn:monomials} do not form an eigenspace basis (an example is the monomial $\chi_{-\frac{3}{2}}|0\rangle$, which is not an eigenvector for $L^t_0$).

There is a fourth grading on $\mathit{F_{\chi}}$,   induced by the yet  another family of Virasoro fields on $\mathit{F_{\chi}}$: for any $a, b \in \mathbb{C}$ the field
\begin{equation}
\label{eqn:Vir1}
L^{h} (z^2)=-\frac{1}{2}:h^{\mathbb{Z}} (z^2)^2: +a \partial_{z^2} h^{\mathbb{Z}} (z^2) +\frac{b}{z^2}h^{\mathbb{Z}}(z^2) +\frac{2ab- b^2}{2z^4}.
\end{equation}
is a Virasoro field with central charge $1+12a^2$, the central charge is independent of $b$ (\cite{AngMLB}).
We have
\begin{equation}
\label{eqn:VirGrading-h}
L^{h}_0 = -\frac{1}{2}(h^{\mathbb{Z}}_0)^2 -\sum_{n\in \mathbb{Z}_{+}} h^{\mathbb{Z}}_{-n}h^{\mathbb{Z}}_n +\left(b-a-an\right)h^{\mathbb{Z}}_{n}.
\end{equation}
For simplicity here we will only consider  the case  $a=b=0$.
$L^h_0$ is also diagonalizable, again due to Theorem \ref{thm:HeisDecomp}, its second part. The monomials from \eqref{eqn:monomials} do not form an eigenspace basis here either: an example is the monomial $\chi_{-\frac{5}{2}}|0\rangle$, which is not an eigenvector for $L^h_0$.
\begin{remark}
The three Virasoro fields are different, neither is part of the family of any of the other two, as we explained in \cite{AngMLB}.
\end{remark}

Since we have multiple gradings, we can consider multiparameter characters (graded dimensions). The potential obstruction to such a multiparameter character is the lack of commuting of the respective grading operators. In our case, the grading operators $L_0$ and $h^{\mathbb{Z}}_0$  commute among themselves. Consequently,  in \cite{AngCKPSecond} we already introduced the  following character with respect to the $L_0$ and $h^{\mathbb{Z}}_0$ grading operators:
\begin{equation}
dim_{q, z}^{L_0, h^{\mathbb{Z}}_0} \mathit{F_{\chi}}:= tr_{\mathit{F_{\chi}}}q^{2L_0}z^{h^{\mathbb{Z}}_0}.
\end{equation}
But, the third grading operator  $L^t_0$ does not commute with the charge grading operator $h^{\mathbb{Z}}_0$, and consequently,  with the 4th grading by $L^h_0$ either:
\begin{lem}\label{lem:nonsiimultdiagon}
The operators  $L^t_0$ and $h^{\mathbb{Z}}_0$ do not commute. Therefore $L^t_0$ and $h^{\mathbb{Z}}_0$ cannot be simultaneously diagonalized. The same is true for  the operators $L^t_0$ and $L^h_0$.
\end{lem}
\begin{proof}
Consider the highest weight vector $v_t = \chi_{-\frac{3}{2}}|0\rangle -\frac{1}{3}\chi_{-\frac{1}{2}}^3|0\rangle$ for the twisted Heisenberg algebra
$\mathcal{H}_{\mathbb{Z}+1/2}$. We have from the definition of $L^t_0$ that on highest weight vectors it acts as 0, and so
\[
L^t_0 \left(  \chi_{-\frac{3}{2}}|0\rangle -\frac{1}{3}\chi_{-\frac{1}{2}}^3|0\rangle \right) = L^t_0 v_t =0.
\]
Therefore $h^{\mathbb{Z}}_0 L^t_0 v_t =0$.
But
\[
h^{\mathbb{Z}}_0 v_t =h^{\mathbb{Z}}_0  \left(  \chi_{-\frac{3}{2}}|0\rangle -\frac{1}{3}\chi_{-\frac{1}{2}}^3|0\rangle \right) = -1\cdot \chi_{-\frac{3}{2}}|0\rangle -3\cdot \frac{1}{3}\chi_{-\frac{1}{2}}^3|0\rangle = - \chi_{-\frac{3}{2}}|0\rangle -\chi_{-\frac{1}{2}}^3|0\rangle.
\]
As we mentioned above, $L^t_0$ is diagonalizable,  but we have to use the Heisenberg decomposition of Theorem \ref{thm:HeisDecomp} to calculate the action of $L^t_0$ on non-highest weight vectors, such as the vector $\left(- \chi_{-\frac{3}{2}}|0\rangle -\chi_{-\frac{1}{2}}^3|0\rangle\right)$. In particular, we have two basis vectors of degree 3, as there are 2 odd partitions of degree 3, namely $(3)$ and $(1, 1, 1)$. Thus we have to use as basis eigenvectors of degree 3 the highest weight vector $v_t$ of degree 3, and the non-highest-weight vector
\[
h_{-\frac{1}{2}} \chi_{-\frac{1}{2}}|0\rangle =\left(\frac{1}{2}\chi_{-\frac{1}{2}}^2 -\chi_{-\frac{3}{2}}\chi_{\frac{1}{2}}\right)\chi_{-\frac{1}{2}}|0\rangle =\frac{1}{2}\chi_{-\frac{1}{2}}^3|0\rangle - \chi_{-\frac{3}{2}}|0\rangle.
 \]
 Note  we used that $ \chi_{-\frac{1}{2}}|0\rangle$ is the single  highest weight vector of degree 1, and we have from
 Proposition \ref{prop:Heis-chi} that
 \[
 h_{-\frac{1}{2}} =\frac{1}{2}\sum_{n\in \mathbb{Z}+1/2}(-1)^{-n-1/2}\chi_{-n-1}\chi_n =\frac{1}{2}\chi_{-\frac{1}{2}}^2 -\chi_{-\frac{3}{2}}\chi_{\frac{1}{2}} + \chi_{-\frac{5}{2}}\chi_{\frac{3}{2}}-\dots .
 \]
 We can write
\begin{align*}
- \chi_{-\frac{3}{2}}|0\rangle -\chi_{-\frac{1}{2}}^3|0\rangle = 7v_t -8h_{-\frac{1}{2}} \chi_{-\frac{1}{2}}|0\rangle.
\end{align*}
And thus since $[h_{\frac{1}{2}}, h_{-\frac{1}{2}}] =-\frac{1}{2}$ we get
\begin{align*}
L^t_0 \left(- \chi_{-\frac{3}{2}}|0\rangle -\chi_{-\frac{1}{2}}^3|0\rangle\right) = L^t_0 \left(7v_t -8h_{-\frac{1}{2}} \chi_{-\frac{1}{2}}|0\rangle\right) = 0 - 4h_{-\frac{1}{2}} \chi_{-\frac{1}{2}}|0\rangle.
\end{align*}
In any case, $L^t_0 h^{\mathbb{Z}}_0 v_t \neq 0$, therefore $L^t_0 h^{\mathbb{Z}}_0\neq  h^{\mathbb{Z}}_0 L^t_0$ \ on $\mathit{F_{\chi}}$. The proof that the operators $L^t_0$ and $L^h_0$ do not commute is similar, and so we omit it.
\end{proof}
Thus we cannot consider a character $dim_{t, z}^{L^t_0, h^{\mathbb{Z}}_0} \mathit{F_{\chi}}$, as such a trace $tr_{\mathit{F_{\chi}}}t^{L^t_0}z^{h^{\mathbb{Z}}_0}$ would require the simultaneous diagonalization of  $L^t_0$ and $h^{\mathbb{Z}}_0$, and so does not exist.
 Consequently,  a simultaneous diagonalization of all 4 grading operators,  including the  $L^t_0$ grading,  is not possible; and  as a result a 4-parameter graded dimension with these 4 grading operators does not exist either.

But we have the following Lemma, which will ensure that the operators $L_0$ and $L^t_0$ can be simultaneously diagonalized:
\begin{lem} (\cite{OrlovLeur})
For any $v\in \mathit{F_{\chi}}$ a homogeneous element of given degree we have
\begin{equation}
\label{eqn:degreechange}
deg (h^t_{-\frac{2n-1}{2}} v) = (2n-1) +deg (v) , \quad \forall \ n\in \mathbb{Z}_+.
\end{equation}
An eigenvector basis for the operator $L_0$ consists of all the highest weight vectors  for the twisted Heisenberg algebra $\mathcal{H}_{\mathbb{Z}+1/2}$, together with the twisted Heisenberg monomials induced from them:
\[
h^t_{-\frac{2n_s-1}{2}}\dots h^t_{-\frac{2n_2-1}{2}}h^t_{-\frac{2n_1-1}{2}}v_t,
\]
where $v_t$  is any highest weight vector  for the twisted Heisenberg algebra $\mathcal{H}_{\mathbb{Z}+1/2}$, $n_1, n_2, \dots n_s \in \mathbb{Z}_+$ are not necessarily different.
\end{lem}
\begin{proof}
By direct observation of the form of $h^t_{-\frac{2n_s-1}{2}}$ by Proposition \ref{prop:Heis-chi}, namely
\[
h^t_{-\frac{2n-1}{2}} =\frac{1}{2}\sum_{\substack{k, l\in \mathbb{Z}+1/2\\ k+l =-2n+1}} (-1)^{-l-1/2}\chi_{k}\chi_l .
\]
\end{proof}

Therefore we can consider the 2 parameter character induced by  the grading operators $L_0$ and $L^t_0$
\begin{equation}
dim_{q, t}^{L_0, L^t_0} \mathit{F_{\chi}}:= tr_{\mathit{F_{\chi}}}q^{2L_0}t^{L^t_0}.
\end{equation}
By applying the Lemma above, and Theorem \ref{thm:HeisDecomp} (recall $\mathcal{ODP}$ stands for the set of all distinct partitions with parts in  $\mathbb{Z} +1/2$), we  get
\begin{prop}\label{prop:q-t-trace}
\begin{equation}
dim_{q, t}^{L_0, L^t_0} \mathit{F_{\chi}}= \sum_{\mathfrak{p}\in \mathcal{ODP}} \frac{q^{|\mathfrak{p}|}}{\prod_{n\geq 1} (1-q^{2n-1}t^{\frac{2n-1}{2}})}
\end{equation}
\end{prop}
\begin{proof}
From the decomposition of Theorem \ref{thm:HeisDecomp},  $\mathit{F_{\chi}}$ is a
a direct sum of irreducible highest weight Heisenberg $\mathcal{H}_{\mathbb{Z}+1/2}$  modules, each isomorphic to
\[
\mathbb{C}[h^t_{-\frac{1}{2}}, h^t_{-\frac{3}{2}}, \dots , h^t_{-\frac{2n-1}{2}}, \dots ] \cdot v_t,
\]
where $v_t$  is a highest weight vector  for  $\mathcal{H}_{\mathbb{Z}+1/2}$.
On any element such as
\[
h^t_{-\frac{2n_s-1}{2}}\dots h^t_{-\frac{2n_2-1}{2}}h^t_{-\frac{2n_1-1}{2}}v_t
\]
here  $n_1, n_2, \dots n_s \in \mathbb{Z}_+$ are not necessarily different, we have
\[
2L_0 h^t_{-\frac{2n_s-1}{2}}\dots h^t_{-\frac{2n_2-1}{2}}h^t_{-\frac{2n_1-1}{2}}v_t = \big( (2n_s -1) +\dots (2n_2 -1) +(2n_1 -1)\big) h^t_{-\frac{2n_s-1}{2}}\dots h^t_{-\frac{2n_2-1}{2}}h^t_{-\frac{2n_1-1}{2}}v_t.
\]
To calculate the action of $L^t_0$ on such an element, we need the following Lemma:
\begin{lem}
The following relations hold:
\begin{equation}
[ L^t_0, h^t_{-\frac{2n-1}{2}} ] = \frac{2n-1}{2}\cdot  h^t_{-\frac{2n-1}{2}}, \quad \forall \ n\in \mathbb{Z},
\end{equation}
\end{lem}
\begin{proof} We can use Wick's Theorem here (the version generalized to twisted vertex algebras, see \cite{ACJ}), and taking into account that the field $ h_\chi^{\mathbb{Z}+1/2} (z)$ really depends only on $z^2$, see Proposition \ref{prop:Heis-chi}, we calculate:
\begin{align*}
L^{\chi} (z^2)h_\chi^{\mathbb{Z}+1/2} (w) & \sim -\frac{1}{2z^2}:h^{\mathbb{Z}+1/2} (z)^2: h_\chi^{\mathbb{Z}+1/2} (w) \sim \frac{1}{z^2}\cdot \frac{z^2 +w^2}{2(z^2 -w^2)^2} h_\chi^{\mathbb{Z}+1/2} (z)\\
& \sim \left(\frac{1}{w^2} -\frac{z^2-w^2}{w^4}+\dots\right) \left(\frac{1}{2(z^2 -w^2)} +\frac{w^2}{(z^2 -w^2)^2}\right)   h_\chi^{\mathbb{Z}+1/2} (z) \\
& \sim \left(-\frac{1}{2w^2(z^2 -w^2)} +\frac{1}{(z^2 -w^2)^2}\right)   h_\chi^{\mathbb{Z}+1/2} (z)\\
& \sim -\frac{1}{2w^2(z^2 -w^2)}  h_\chi^{\mathbb{Z}+1/2} (w) +\frac{1}{(z^2 -w^2)^2}  h_\chi^{\mathbb{Z}+1/2} (w) + \frac{1}{z^2 -w^2}  \partial_{w^2} h_\chi^{\mathbb{Z}+1/2} (w).
\end{align*}
And so collecting the term at $z^{-4}$ we get
\begin{equation*}
[ L^t_0,\  h_\chi^{\mathbb{Z}+1/2} (w)] = -\frac{1}{2} h_\chi^{\mathbb{Z}+1/2} (w) +   h_\chi^{\mathbb{Z}+1/2} (w) + w^2 \partial_{w^2} h_\chi^{\mathbb{Z}+1/2} (w).
\end{equation*}
In terms of modes we have
\[
\sum_{n\in \mathbb{Z}+1/2} [ L^t_0, h^t_n w^{-2n-1}] = \frac{1}{2}\sum_{n\in \mathbb{Z}+1/2}  h^t_n w^{-2n-1} +\frac{-2n-1}{2}\sum_{n\in \mathbb{Z}+1/2}  h^t_n (w^2)^{\frac{-2n-1}{2}};
\]
thus
\[
[ L^t_0, h^t_n ] =-n h^t_n, \quad \text{for\ any}\quad n\in \mathbb{Z}+1/2.
\]
\end{proof}
Now this allows us to go back to the proof of the proposition, as this shows  that
\[
L^t_0 h^t_{-\frac{2n_s-1}{2}}\dots h^t_{-\frac{2n_2-1}{2}}h^t_{-\frac{2n_1-1}{2}}v_t = \frac{1}{2}\big( (2n_s -1) +\dots (2n_2 -1) +(2n_1 -1)\big) h^t_{-\frac{2n_s-1}{2}}\dots h^t_{-\frac{2n_2-1}{2}}h^t_{-\frac{2n_1-1}{2}}v_t.
\]
Thus the graded dimension of the irreducible module $\mathbb{C}[h^t_{-\frac{1}{2}}, h^t_{-\frac{3}{2}}, \dots , h^t_{-\frac{2n-1}{2}}, \dots ] \cdot v_t$ is
\[
dim_{q, t}^{L_0, L^t_0} \mathbb{C}[h^t_{-\frac{1}{2}}, h^t_{-\frac{3}{2}}, \dots , h^t_{-\frac{2n-1}{2}}, \dots ] \cdot v_t =\frac{q^{deg(v_t)}}{\prod_{n\geq 1} (1-q^{2n-1}t^{\frac{2n-1}{2}})}.
\]
Hence the proof is completed by  Theorem \ref{thm:HeisDecomp}, as the highest weight vectors $v_t$ in the decomposition are indexed by $\mathcal{ODP}$.
\end{proof}
\begin{remark}
Setting $t=1$ we obtain the not terribly interesting result that
\begin{equation}
dim_{q}^{L_0} \mathit{F_{\chi}} = \sum_{\mathfrak{p}\in \mathcal{ODP}} \frac{q^{|\mathfrak{p}|}}{\prod_{n\geq 1} (1-q^{2n-1})}.
\end{equation}
This result is not surprising at all, as of course we have
\[
\sum_{\mathfrak{p}\in \mathcal{ODP}} q^{|\mathfrak{p}|} = \prod_{n\geq 1} (1+q^{\frac{2n-1}{2}}),
\]
and thus
\[
dim_{q}^{L_0} \mathit{F_{\chi}} =\sum_{\mathfrak{p}\in \mathcal{ODP}} \frac{q^{|\mathfrak{p}|}}{\prod_{n\geq 1} (1-q^{2n-1})} =\frac{ \prod_{n\geq 1} (1+q^{\frac{2n-1}{2}})}{\prod_{n\geq 1} (1-q^{2n-1})} =\frac{1}{ \prod_{n\geq 1} (1-q^{\frac{2n-1}{2}})},
\]
which we knew  already.
\end{remark}

Since $L^t_0$ doesn't commute with any of the other two grading operators, we now exclude $L^t_0$ and look at
the combinations between the other three grading operators. Ultimately, we want to derive
the tri-parameter graded dimension
\begin{equation}
dim_{q, z, r}^{L_0, h^{\mathbb{Z}}_0, L^h_0} \mathit{F_{\chi}}:= tr_{\mathit{F_{\chi}}}q^{2L_0}z^{h^{\mathbb{Z}}_0}r^{L^h_0}.
\end{equation}
For this graded dimension we will need to introduce some additional partition notation. In particular, we want to make the connection between this character and the birank of a bipartition (two-colored partition), introduced in \cite{HL}, see also \cite{Garvan}.
In fact, there are two important sets of bipartitions we need to consider.
The first set of bipartitions is the set $\mathfrak{BP}_{HI}$, consisting of bipartitions
$(\pi_1 \ | \  \pi_2)$, such that $\pi_1$ and $\pi_2$ are partitions with parts in $\mathbb{Z}_+ +1/2$.

To each   each monomial $\left(\chi _{-j_k}\right)^{m_k}\dots \left(\chi _{-j_2}\right)^{m_2}\left(\chi _{-j_1}\right)^{m_1}|0\rangle$ from the basis  \eqref{eqn:monomials} of $\mathit{F_\chi}$ we can assign a bipartition $(\pi_1 \ | \  \pi_2)\in \mathfrak{BP}_{HI}$ as follows:  the partition $\pi_1$ would consists of the $j_s \in 2\mathbb{Z} -1/2$, and $\pi_2$ would consists of the $j_s \in 2\mathbb{Z} +1/2$.
\begin{example}
To the monomial $\chi_{-\frac{9}{2}}^3\chi_{-\frac{5}{2}}^2\chi_{-\frac{3}{2}}\chi_{-\frac{1}{2}} |0\rangle$ we associate the bipartition  $(\pi_1 \ | \  \pi_2)$ with
\[
 \pi_1 = \left( \frac{3}{2}\right), \quad \pi_2 =\left(\frac{9}{2},\  \frac{9}{2},\  \frac{9}{2}, \ \frac{5}{2}, \ \frac{5}{2},\  \frac{1}{2}\right).
\]
We can also write the partition $\pi_2$ in the notation $\pi_1 =\left(\left(\frac{9}{2}\right)^3, \left(\frac{5}{2}\right)^3, \left(\frac{1}{2}\right)^1\right)$.
\end{example}
\begin{notation}
Denote by $\# \pi$  the number of nonzero parts in the partition $\pi$.   $|(\pi_1 \ | \  \pi_2)|$ will denote the weight of the bipartition $(\pi_1 \ | \  \pi_2)$, i.e.,
$|(\pi_1 \ | \  \pi_2)|:= |\pi_1 | + |\pi_2|$.
\end{notation}
In the example above $\# \pi_1 =1$, $\# \pi_2 =6$, $|(\pi_1 \ | \  \pi_2)| =3\cdot \frac{9}{2} +3\cdot \frac{5}{2} +\frac{1}{2} +\frac{3}{2}$.

\begin{defn} (\cite{HL},  \cite{Garvan})
 The Hammond-Lewis birank of a bipartition $(\pi_1 \ | \  \pi_2)$ is  defined to be
 \begin{equation}
 birank =\# \pi_2 -\# \pi_1.
 \end{equation}
 \end{defn}
 \begin{remark}
 Observe that the Hammond-Lewis birank  of a bipartition $(\pi_1 \ | \  \pi_2)$ corresponding to a monomial in the basis \eqref{eqn:monomials} is precisely the charge of the monomial.
 \end{remark}
\begin{remark}
Originally, the Hammond-Lewis birank was  defined  for bipartitions with integer parts, with a (immaterial) minus sign difference . We extended the definition as above because of the following consideration:
 In \cite{AngMLB} we proved that  the field  $\chi(z)$,  and its descendant field $\chi(-z)$, generate a twisted vertex algebra on the space $\mathit{F_\chi}$, which is twisted-vertex-algebra-isomorphic to the $\beta-\gamma$ system and its Fock space (i.e., with singularities both at $z=w$ and $z=-w$ formally allowed). If we consider then the isomorphic $\beta-\gamma$ system,  the monomials corresponding to the basis \eqref{eqn:monomials} are indexed there with  an integer indexing set. And so if we use the image of the  untwisted Heisenberg field  in the  $\beta-\gamma$ system, and  the image of the monomial  basis \eqref{eqn:monomials}, we would get a usual type of  bipartition assigned to each monomial-- a bipartition with integer parts, and the charge of the monomial will correspond exactly to the originally defined Hammond-Lewis birank as in \cite{HL}, \cite{Garvan}.
\end{remark}
Consider now the subspace of the highest weight vectors $\mathit{F^{hwv}_\chi}$ of $\mathit{F_\chi}$.  We know that it has a basis indexed by $\mathfrak{P}_{tdo}$, where  $\mathfrak{P}_{tdo}$ denoted the set  of distinct partitions with the first part a triangular number, and the other parts being distinct half-integers from $\mathbb{Z}+1/2$. Namely, to each partition from $\mathfrak{P}_{tdo}$ we can assign a highest weight vector, due to the Heisenberg decomposition Theorem \ref{thm:HeisDecomp}. But the basis of $\mathit{F^{hwv}_\chi}$ can be indexed also in another way:

\begin{defn} Denote by $\mathfrak{BP}_{DI}$ the set of distinct integer bipartitions, namely bipartitions  $(\pi_1 \ | \  \pi_2)$ such that each of  $\pi_1$ and $\pi_2$ is a partition with distinct integer parts:
$\pi_1 =(m_k, \dots, m_2, m_1), \ \pi_2 = (n_s, \dots , n_2, n_1)$, where $m_k>\dots m_2 >m_1;  \ n_s>\dots n_2>n_1$, \   $m_i, n_j \in \mathbb{Z}_{>0}$, for any  \ $i=1, 2, \dots, k; j=1, 2, \dots s$.
Note that the same integer is allowed to occur in both $\pi_1$ and $\pi_2$.
\end{defn}
According to Corollary \ref{cor:FDecomp}, since the highest weight vectors can be considered to be elements of $\mathit{SF}$, we can assign a highest weight vector to each distinct integer bipartition $(\pi_1 \ | \  \pi_2)\in \mathfrak{BP}_{DI}$, as follows: if
\[
(\pi_1 \ | \  \pi_2) =\big( (m_k, \dots, m_2, m_1) \ | \ (n_s, \dots , n_2, n_1)\big),
\]
where $m_k>\dots m_2 >m_1;  \ n_s>\dots n_2>n_1; \ m_i, n_j \in \mathbb{Z}_{>0}, \ i=1, 2, \dots, k; j=1, 2, \dots s$,
then we assign to  $(\pi_1 \ | \  \pi_2)$ the highest weight vector
\begin{align*}
H^{\beta}_{(-m_k)}\dots H^{\beta}_{(-m_2)}&H^{\beta}_{(-m_1)}H^{\gamma}_{(-n_s)}\dots H^{\gamma}_{(-n_2)}H^{\gamma}_{(-n_1)}|0\rangle .
\end{align*}
\begin{prop}\label{prop:counting-formula}
There is a one-to-one correspondence between the set $\mathfrak{P}_{tdo}$ and the set of the distinct integer bipartitions $\mathfrak{BP}_{DI}$, via the basis of $\mathit{F^{hwv}_\chi}$:
\[
\mathfrak{P}_{tdo}   \xleftrightarrow[\ \ \ \ \ ]{\sim} \quad \text{basis \ of}\  \mathit{F^{hwv}_\chi} \quad  \xleftrightarrow[\ \ \ \ \ ]{\sim} \quad  \mathfrak{BP}_{DI}.
\]
Further,
the Hammond-Lewis birank  of a distinct integer bipartition $(\pi_1 \ | \  \pi_2)\in \mathfrak{BP}_{DI}$ is precisely the charge of the highest weight vector corresponding uniquely to the bipartition $(\pi_1 \ | \  \pi_2)\in \mathfrak{BP}_{DI}$.

This one-to-one correspondence is not an isometry, instead  for
$\mathfrak{P}_{tdo} \ni \mathfrak{p}   \longleftrightarrow (\pi_1 \ | \  \pi_2)\in \mathfrak{BP}_{DI}$, the  weights $|\mathfrak{p}|$ and $|(\pi_1 \ | \  \pi_2)|$ are connected via the following function $\mathcal{W}:\mathfrak{BP}_{DI}\to \mathbb{Z}_{\geq 0}+\frac{1}{2}$:
\begin{equation}\label{eqn:weightformula}
|\mathfrak{p}|  = \mathcal{W}\left((\pi_1 \ | \  \pi_2)\right) : = 2|(\pi_1 \ | \  \pi_2)| + 2ns - \frac{n(2n-1)}{2} - \frac{s(2s+1)}{2},
\end{equation}
Here $n= \# \pi_1$,   the number of nonzero parts in the partition $\pi_1$,  $s= \# \pi_2$,   the number of nonzero parts in the partition $\pi_2$,  $T_n$ (correspondingly $T_s$) is the $n$-th (correspondingly $s$-th) triangular number.
\end{prop}

The first part of the proposition follows directly as a corollary to Theorem \ref{thm:HeisDecomp} and Corollary \ref{cor:FDecomp}.
Formula \eqref{eqn:weightformula}  will hinge on the following Lemma that derives the degree of the highest weight vector corresponding to a given  distinct integer bipartition.
\begin{lem}\label{lem:degree-calculation}
 I. Let  $v\in \mathit{F^{hwv}_\chi}$ corresponds to a bipartition $(  \pi_1 \ | \   \emptyset )$, where  $\pi_2$ is the empty partition, and $\pi_1 =(m_n, \dots, m_2, m_1)$,  $m_n>\dots m_2 >m_1$, \   $m_i \in \mathbb{Z}_{>0}$, for any  \ $i=1, 2, \dots, n$. Then
\begin{equation}\label{eqn:degree1}
deg(v) = \frac{3n}{2} + 2|\pi_1|  -2T_n = 2|\pi_1| - \frac{n(2n-1)}{2}.
\end{equation}
Here $n= \# \pi_1$,   the number of nonzero parts in the partition $\pi_1$,   $T_n$ is the $n$-th triangular number.
 \\
II. Let  $v\in \mathit{F^{hwv}_\chi}$ corresponds to a bipartition $(\emptyset \ | \  \pi_2   )$, where  $\pi_1$ is the empty partition, and $\pi_2 =(m_n, \dots, m_2, m_1)$,  $m_n>\dots m_2 >m_1$, \   $m_i \in \mathbb{Z}_{>0}$, for any  \ $i=1, 2, \dots, n$. Then
\begin{equation}\label{eqn:degree2}
deg(v) = \frac{n}{2} + 2|\pi_1|  -2T_n = 2|\pi_1| - \frac{n(2n+1)}{2}.
\end{equation}
Here $n= \# \pi_2$,  $T_n$ is the $n$-th triangular number.\\
III. Let  $v\in \mathit{F^{hwv}_\chi}$ corresponds to a bipartition $(  \pi_1 \ | \  \pi_2 )$, where $\pi_1 =(m_n, \dots, m_2, m_1)$,  $m_n>\dots m_2 >m_1$, \   $m_i \in \mathbb{Z}_{>0}$, for any  \ $i=1, 2, \dots, n$; and   $\pi_2 =(l_s, \dots, l_2, l_1)$,  $l_s>\dots l_2 >l_1$, \   $l_i \in \mathbb{Z}_{>0}$, for any  \ $i=1, 2, \dots, s$. Then
\begin{equation}\label{eqn:degree3}
deg(v) = \frac{3n}{2} + 2|\pi_1|  -2T_n +2ns +  \frac{s}{2} + 2|\pi_2|  -2T_s= 2|\pi_1| + 2|\pi_2| + 2ns - \frac{n(2n-1)}{2} - \frac{s(2s+1)}{2},
\end{equation}
Here $n= \# \pi_1$,   the number of nonzero parts in the partition $\pi_1$,  $s= \# \pi_2$,   the number of nonzero parts in the partition $\pi_2$,  $T_n$ (correspondingly $T_s$) is the $n$-th (correspondingly $s$-th) triangular number.
\end{lem}
The cases I and II above are actually special cases of III (when $s$ =0, or correspondingly $n$=0), but we chose to state the formulas on their own for clarity.
\begin{proof}
We will look  at  three types of  representative examples which will illustrate  two important principles when counting the degrees,  without more of the already excessive indexing.
We start with the basic example of action on the vacuum vector $|0\rangle$. By using directly the  formulas  \eqref{eqn:Hbeta-Hgamma}  for the fields   $H^{\beta}(z^2)$ and $H^{\gamma}(z^2)$ we can calculate
\begin{align*}
H^{\beta} (z^2) |0\rangle &= H^{\beta}_{(-1)}|0\rangle +H^{\beta}_{(-2)}|0\rangle z^2 +\dots +H^{\beta}_{(-k)}|0\rangle z^{2k-2} +\dots \\
& =  V^+(z)^{-1}\beta_\chi(z^2) |0\rangle =\big(1-h_{-1}^{\mathbb{Z}} z^2 +\dots\big)\big(\chi _{-3/2}|0\rangle + \chi _{-7/2}|0\rangle z^2 +\dots \big)\\
&=\chi _{-3/2}|0\rangle +\big(2\chi _{-7/2}|0\rangle -\chi _{-3/2}^2\chi _{-1/2}|0\rangle\big) z^2 +\dots
\end{align*}
Thus we see that acting on $|0\rangle$,  $H^{\beta}_{(-1)}$ raises the degree by $\frac{3}{2}$,  $H^{\beta}_{(-2)}$ raises the degree by $\frac{7}{2}$, and in general, acting on $|0\rangle$  $H^{\beta}_{(-k)}$ raises the degree by $2k-\frac{1}{2} $. We have
\[
deg\big(H^{\beta}_{(-1)}|0\rangle\big) =\frac{3}{2}, \quad deg\big(H^{\beta}_{(-2)}|0\rangle\big) =\frac{7}{2}=\frac{3}{2} +2, \quad \dots, \quad  deg\big(H^{\beta}_{(-k)}|0\rangle\big) =2k-\frac{1}{2}=\frac{3}{2} +2(k-1).
\]
Similarly, acting on $|0\rangle$,   $H^{\gamma}_{(-n)}$ raises the degree by $2n-\frac{3}{2}$.

But, on elements of other charges this is not  the case, and so next we consider the example of charge -1:
\begin{align*}
H^{\beta} (z^2)\big(\chi _{-3/2} |0\rangle\big) &= H^{\beta} (z^2)\big(H^{\beta}_{(-1)} |0\rangle\big) \\
&=H^{\beta}_{(-2)}\big(H^{\beta}_{(-1)} |0\rangle\big) z^2 +H^{\beta}_{(-3)}\big(H^{\beta}_{(-1)} |0\rangle\big) z^4 +\dots +H^{\beta}_{(-k)}H^{\beta}_{(-1)}|0\rangle z^{2k-2} +\dots \\
&=\chi _{-3/2}^2 |0\rangle z^2 + \big(\chi _{-7/2}\chi _{-3/2}|0\rangle -h^{\mathbb{Z}}_{-1}\chi _{-3/2}^2|0\rangle\big) z^4 +\dots.
\end{align*}
So we see that acting on $\chi _{-3/2} |0\rangle =H^{\beta}_{(-1)}|0\rangle$,   $H^{\beta}_{(-1)}$ acts as 0 (as  it is fermionic),  but instead $H^{\beta}_{(-2)}$ raises the degree by $\frac{3}{2}$ (and not as before by $\frac{7}{2}$), and in general, acting on $\chi _{-3/2}|0\rangle$,   $H^{\beta}_{(-k)}, \  k\geq 2, $ raises the degree by $2k-\frac{5}{2}$. Thus we have a "resetting of the counter" principle, where each  $H^{\beta}_{(-k)}$ adds a  degree contribution depending on the charge of the vector it is acting on (or equivalently, the number of nonzero partition parts in the bipartition describing that highest weight vector). Specifically, we have
\begin{align*}
&deg\big(H^{\beta}_{(-2)}H^{\beta}_{(-1)}|0\rangle\big) =\frac{2\cdot 3}{2}, \quad deg\big(H^{\beta}_{(-3)}H^{\beta}_{(-1)}|0\rangle\big) =\frac{7}{2} +\frac{3}{2} =\frac{2\cdot 3}{2} +2,\\
& deg\big(H^{\beta}_{(-k)}H^{\beta}_{(-1)}|0\rangle\big) =\frac{3}{2} + 2k-\frac{5}{2}=\frac{2\cdot 3}{2} +2(k-2).
\end{align*}
"Resetting of the counter" refers to the fact that  when we considering adding new contributions   to the degree, we always start at the lowest possible addition (in the example above,  we started at $\frac{3}{2}$, even though it was  $H^{\beta}_{(-2)}$ that acted).

At each fixed charge $-n<0$ (here $n>0$),  the vector $\chi _{-3/2}^n |0\rangle$ is the vector with minimal degree ,  $\frac{3n}{2}$,  with that charge $-n$. This vector is given by
 \begin{equation}
\chi_{-\frac{3}{2}}^n|0\rangle = H^{\beta}_{(-n)}\dots H^{\beta}_{(-2)} H^{\beta}_{(-1)}|0\rangle,
\end{equation}
and  corresponds to the standard distinct partition with $n$ parts $(n, n-1, n-2, \dots 1)$. Acting on $\chi _{-3/2}^n |0\rangle$, the first nonzero possible action is by $H^{\beta}_{(-n-1)}$, and that will add  $\frac{3}{2}$ to the degree, by the "resetting of the counter" principle.

Let us consider the highest weight vector
\[
v = H^{\beta}_{(-m_n)}\dots H^{\beta}_{(-m_2)}\dots H^{\beta}_{(-m_1)}|0\rangle,
\]
corresponding to the bipartition $(  \emptyset \ | \  \pi_2 )$, where  $\pi_1$ is the empty partition, and $\pi_2 =(m_k, \dots, m_2, m_1)$.  The same  "resetting of the counter" principle as above applies. And so to calculate its degree we will have to see "how much" this vector differs from the vector with minimal degree of that same charge  (the charge of that vector equals  $-n$, with $n$ the number of nonzero parts in the partition $\pi_2$). Thus we see that the degree  is
\[
deg(v) = \frac{3n}{2} + 2(m_n -n) + \dots + 2(m_2-2) + 2(m_1-1),
\]
which gives  the formula \eqref{eqn:degree1}. Formula \eqref{eqn:degree2} is proved the same way.

To  prove the formula \eqref{eqn:degree3}, we need to discuss the second important principle when counting the degrees of the highest weight vectors. Consider the following example, it is
  the easiest example of a "mixture" between the $\beta$ and $\gamma$ fields, i.e., neither the $\pi_1$, nor the $\pi_2$ are the empty partitions:
\begin{align*}
H^{\beta} (z^2)\big(\chi _{-1/2} |0\rangle\big)
&=H^{\beta} (z^2)H^{\gamma}_{(-1)} |0\rangle\big) \\
&= H^{\beta}_{(1)}H^{\gamma}_{(-1)} |0\rangle\cdot \frac{1}{z^4} +H^{\beta}_{(0)}H^{\gamma}_{(-1)} |0\rangle\cdot \frac{1}{z^2} +H^{\beta}_{(-1)}H^{\gamma}_{(-1)} |0\rangle  +\dots +H^{\beta}_{(-k)}H^{\gamma}_{(-1)} |0\rangle z^{2k-2} +\dots \\
&=\frac{|0\rangle}{z^4} +0\cdot \frac{1}{z^2} + v_{4;  0}  + v_{6;  0}z^2 + \dots  + v_{2k+4;  0}z^{2k} +\dots ,
\end{align*}
By direct comparison with the formula \eqref{eqn:Hbeta-Hgamma}:
\begin{align*}
H^{\beta} (z^2)\big(\chi _{-1/2} |0\rangle\big) &=  V^+(z)^{-1}\beta_\chi(z^2) \cdot \frac{1}{z^2} \big(\chi _{-1/2} |0\rangle\big) \\
& =  \frac{1}{z^2} \big(1-h_{-1}^{\mathbb{Z}} z^2 +\dots\big)\Big(\frac{ |0\rangle}{z^2} +\chi _{-3/2}\chi _{-1/2} |0\rangle + \chi _{-7/2}\chi _{-1/2}|0\rangle z^2 +\dots \Big)\\
&= \frac{|0\rangle}{z^4} + \big(-h_{-1}^{\mathbb{Z}}|0\rangle +\chi _{-3/2}\chi _{-1/2} |0\rangle\big)\cdot \frac{1}{z^2} +\\
&\hspace{2cm} + \Big(\frac{1}{2}\big((h_{-1}^{\mathbb{Z}})^2 -h_{-2}^{\mathbb{Z}}\big)|0\rangle -h_{-1}^{\mathbb{Z}}\chi _{-3/2}\chi _{-1/2} |0\rangle  + \chi _{-7/2}\chi _{-1/2}|0\rangle\Big) +\dots,
\end{align*}
 we see that  the vector $v_{4;  0}$ is of charge 0  and degree 4. One can calculate that it  is given by
\[
v_{4;  0} =\chi_{-\frac{3}{2}}^2 \chi_{-\frac{1}{2}}^2|0\rangle  -2\chi_{-\frac{7}{2}}\chi_{-\frac{1}{2}} |0\rangle +2\chi_{-\frac{5}{2}}\chi_{-\frac{5}{2}}|0\rangle.
\]
More generally, in the above formula $v_{2k+ 4;  0}$ is of charge 0 and degree $2k+ 4$. Thus we see that action of $H^{\beta}_{(1)}$ on $H^{\gamma}_{(-1)} |0\rangle =\chi _{-1/2} |0\rangle$ lowers the degree by $\frac{1}{2}$, $H^{\beta}_{0}$ acts as 0, $H^{\beta}_{(-1)}$  raises the degree by $\frac{7}{2}$,  $H^{\beta}_{(-2)}$ raises the degree by $\frac{11}{2}$, and in general, acting on $\chi _{-1/2} |0\rangle$  $H^{\beta}_{(-k)}$ raises the degree by $2k+\frac{3}{2}$. Now if the partition $\pi_2$ was empty, the  $H^{\beta}_{(-1)}$  would have raised the degree by $\frac{3}{2}$, but instead, it now raised the degree by $\frac{7}{2}$. This is due to the  positively indexed $H^{\beta}_{(1)}$ acting nontrivially, and thus the presence of the  $H^{\beta}_{0}$, which always acts as 0, but nominally  should have raised the degree by $\frac{3}{2}$. That forces the $H^{\beta}_{(-1)}$  to "skip"  $\frac{3}{2}$, but instead go to the next available  $\frac{7}{2} =\frac{3}{2} +2\cdot 1$.

Similarly, consider one more example of the same kind:
\begin{align*}
H^{\beta} (z^2&)\big(\chi _{-1/2}^2 |0\rangle\big)
= H^{\beta} (z^2)\big(H^{\gamma}_{(-2)}H^{\gamma}_{(-1)} |0\rangle\big) \\
&= H^{\beta}_{(2)}\big(H^{\gamma}_{(-2)}H^{\gamma}_{(-1)} |0\rangle\big)\frac{1}{z^6} +H^{\beta}_{(1)}\big(H^{\gamma}_{(-2)}H^{\gamma}_{(-1)} |0\rangle\big)\frac{1}{z^4}  +H^{\beta}_{(0)}\big(\chi _{-1/2}^2 |0\rangle\big)\frac{1}{z^2} +H^{\beta}_{(-1)}\big(\chi _{-1/2}^2 |0\rangle\big)  +\dots  \\
&=\frac{H^{\gamma}_{(-1)} |0\rangle}{z^6} + \frac{-H^{\gamma}_{(-2)} |0\rangle}{z^4} +0\cdot \frac{1}{z^2} + v_{\frac{13}{2};  0}  + v_{\frac{17}{2};  0}z^2 + \dots  +  \\
&= \frac{\chi _{-1/2}|0\rangle}{z^6} - \frac{2\chi _{-5/2}|0\rangle +\chi _{-3/2}\chi _{-1/2}^2|0\rangle}{z^4} + 0\cdot \frac{1}{z^2}  +  v_{\frac{13}{2};  0} + \dots
\end{align*}
Hence action of $H^{\beta}_{(2)}$ on $\chi _{-1/2}^2 |0\rangle$ lowers the degree by $\frac{1}{2}$,  $H^{\beta}_{(1)}$ raises the degree by $\frac{3}{2}$,  $H^{\beta}_{0}$ acts as 0 (but nominally  should have raised the degree by $\frac{7}{2}$), $H^{\beta}_{(-1)}$  raises the degree by $\frac{11}{2}$,  $H^{\beta}_{(-2)}$ raises the degree by $\frac{15}{2}$, and in general, acting on $\chi _{-1/2}^2 |0\rangle$,  $H^{\beta}_{(-k)}$ raises the degree by $2k+\frac{7}{2}$. Hence, we again see that $H^{\beta}_{(-1)}$  which would have raised the degree by $\frac{3}{2}$,   now raised the degree instead by $\frac{11}{2}$. This is due to the fact that since the charge of $\chi _{-1/2}^2 |0\rangle$ is 2,  $H^{\beta}_{(-1)}$ now has to "skip" two steps, from  $\frac{3}{2}$ to   $\frac{11}{2}=\frac{3}{2} +2\cdot2$.
One can easily see that the "skipping $s$ steps" principle will hold when $H^{\beta}_{(-1)}$ acts on any vector of charge $s$.

Consider then the bipartition $(\pi_1 \ | \  \pi_2)\in \mathfrak{BP}_{DI}$
\[
(\pi_1 \ | \  \pi_2) =\big( (m_k, \dots, m_2, m_1) \ | \ (l_s, \dots , l_2, l_1)\big),
\]
where $m_k>\dots m_2 >m_1;  \ l_s>\dots l_2>l_1; \ m_i, l_j \in \mathbb{Z}_{>0}, \ i=1, 2, \dots, k; j=1, 2, \dots s$,
and the highest weight vector $v$ assigned to it:
\begin{align*}
v= H^{\beta}_{(-m_k)}\dots H^{\beta}_{(-m_2)}&H^{\beta}_{(-m_1)}H^{\gamma}_{(-l_s)}\dots H^{\gamma}_{(-l_2)}H^{\gamma}_{(-l_1)}|0\rangle .
\end{align*}
Applying  both the principle of "skipping" $s$ steps (because the charge of the vector $H^{\gamma}_{(-l_s)}\dots H^{\gamma}_{(-l_2)}H^{\gamma}_{(-l_1)}|0\rangle$ is $s$), and then the  "resetting the counter" principle, we see that
\[
deg(v) = n\left(\frac{3}{2} +2s\right) + 2(m_n -n) + \dots + 2(m_2-2) + 2(m_1-1) + \frac{s}{2} + 2(l_s -s) + 2\dots (l_2-2) + 2(l_1-1),
\]
which gives  the formula \eqref{eqn:degree3}.
\end{proof}
The  Proposition \ref{prop:counting-formula} then follows directly.
This leads to the following:
\begin{thm}\label{thm:triplecharacter}
\begin{equation}
dim_{q, z, r}^{L_0, h^{\mathbb{Z}}_0, L^h_0} \mathit{F_{\chi}}= \sum_{(\pi_1  | \pi_2) \in \mathfrak{BP}_{DI}} \frac{q^{\mathcal{W}\left((\pi_1 | \pi_2)\right)}z^{birank\left((\pi_1 | \pi_2)\right)} r^{-birank\left((\pi_1 | \pi_2)\right)^2/2}}{\prod_{n\geq 1} (1-q^{2n}r^n)}.
\end{equation}
Here $\mathcal{W}\left((\pi_1 | \pi_2)\right)$ is the weight function from \eqref{eqn:weightformula}.
\end{thm}
\begin{proof}
The proof of this theorem is similar to the proof of  Proposition \ref{prop:q-t-trace}, and hinges on the fact that
\begin{equation}
[ L^h_0, h^{\mathbb{Z}}_n ] =-n h^{\mathbb{Z}}_n, \quad \text{for\ any\ }\ n\in \mathbb{Z},
\end{equation}
as well as  \eqref{eqn:VirGrading-h}, and the Lemma above.
\end{proof}

An important consequence of any bosonization is that  by calculating the character (graded dimension) of both the fermionic and the bosonic side of the correspondence one can obtain identities relating certain product formulas to certain sum formulas. Such a sum-vs-product  identity perfectly illustrates the equality between the fermionic side (the product formulas) and the bosonic side (the sum formulas). For example, from   the classical boson-fermion correspondence (of type A) one can directly obtain the Jacobi triple product identity --- as was done in \cite{Kac}, and in \cite{AngD-A}  for the bosonization of type D-A. To derive the relevant sum-vs-product formula for the CKP,
we continue by considering just the first two grading operators, $h^{\mathbb{Z}}_0$ and $L_0$, as they are in some sense the most natural grading operators for the Fock space $ \mathit{F_{\chi}}$. In the CKP case, as we noticed, nothing is quite as straightforward as in the usual boson-fermion correspondence case.  Here the "fermionic side" is not purely fermionic, but instead one can view it as a fermion times a boson (see Corollary \ref{cor:FDecomp}) . In  \cite{AngCKPSecond} we showed the bosonic side of the character:
\[
dim_{q, z}^{L_0, h^{\mathbb{Z}}_0} \mathit{F_{\chi}} =\frac{1}{\prod_{j\in \mathbb{Z}_{+}} \big(1-zq^{2j- \frac{3}{2}}\big) \big(1-z^{-1}q^{2j- \frac{1}{2}}\big)}
\]
Here we obtain the new (and rather complicated formulas):
\begin{thm}\label{thm:sum-vs-product}
\begin{equation}\label{eqn:characterF-hwv}
dim_{q, z}^{L_0, h^{\mathbb{Z}}_0} \mathit{F^{hwv}_{\chi}} =\frac{\prod_{l=1}^{\infty} \big(1+zq^{2l-\frac{3}{2}}\big)  \prod_{l=1}^{\infty} \big(1+z^{-1}q^{2l-\frac{1}{2}}\big)}{2 \prod_{i=1}^{\infty} (1-q^{4i})(1+q^{2i})}\cdot \left(\sum_{k=0}^{\infty} \frac{z^kq^{\frac{k}{2}} }{1+z^{-1}q^{2k+\frac{3}{2}}} + \sum_{k=0}^{\infty} \frac{z^{-k}q^{\frac{3k}{2}} }{1+zq^{2k+\frac{1}{2}}}\right).
\end{equation}
\begin{equation}\label{eqn:charform-ferm}
dim_{q, z}^{L_0, h^{\mathbb{Z}}_0} \mathit{F_{\chi}} =\frac{\prod_{l=1}^{\infty} \big(1+zq^{2l-\frac{3}{2}}\big)  \prod_{l=1}^{\infty} \big(1+z^{-1}q^{2l-\frac{1}{2}}\big)}{2\prod_{l=1}^{\infty} \big(1-q^{4l}\big)^2}\cdot \left(\sum_{k=0}^{\infty} \frac{z^kq^{\frac{k}{2}} }{1+z^{-1}q^{2k+\frac{3}{2}}} + \sum_{k=0}^{\infty} \frac{z^{-k}q^{\frac{3k}{2}} }{1+zq^{2k+\frac{1}{2}}}\right)
\end{equation}
Comparing the two formulas for $dim_{q, z}^{L_0, h^{\mathbb{Z}}_0} \mathit{F_{\chi}}$ we obtain the identity
\begin{equation}\label{eqn:IdentityR}
\sum_{k=0}^{\infty} \frac{z^kq^{\frac{k}{2}} }{1+z^{-1}q^{2k+\frac{3}{2}}} + \sum_{k=0}^{\infty} \frac{z^{-k}q^{\frac{3k}{2}} }{1+zq^{2k+\frac{1}{2}}} =\frac{2\prod_{l=1}^{\infty} \big(1-q^{4l})^2}{\prod_{l=1}^{\infty} \big(1-z^2q^{4l-3}\big)  \prod_{l=1}^{\infty} \big(1-z^{-2}q^{4l-1}\big)}
\end{equation}
\end{thm}
\begin{remark}
The character \eqref{eqn:characterF-hwv}, besides being necessary for deriving \eqref{eqn:charform-ferm}, is interesting on its own, because the vector space $\mathit{F^{hwv}_\chi}$ can be identified as the vector space of what  we are told  physicists would  regard as the coset symmetry algebra of the $\beta-\gamma$ system by the Heisenberg field $h(z)$, see \cite{Ridout} and \cite{Creutzig-Ridout} (as opposed to the coset vertex subalgebra as considered by mathematicians, see e.g. \cite{Kac}, \cite{LiLep}, which would be the charge 0 subspace of  $\mathit{F^{hwv}_\chi}$ here).
\end{remark}
\begin{remark}
The last identity \eqref{eqn:IdentityR} is somewhat surprising, as the right-hand side manifestly only includes integral powers of $q$ and even powers of $z$, which implies that the left-hand side must do too (not obvious from  its form).  We show in the Appendix that this identity can be obtained also as a specialization of the Ramanujan Psi summation formula (see e.g. \cite{hardy1940ramanujan}, \cite{AndrewsRam-1}, \cite{AndrewsRam})
\[
\sum_{k=-\infty}^{\infty} \frac{(a; q)_k}{(b; q)_k}x^k = \frac{\left(q; q\right)_{\infty}\left(\frac{b}{a}; q\right)_{\infty}\left(\frac{q}{ax}; q\right)_{\infty}\left(ax; q\right)_{\infty}}{\left(b; q\right)_{\infty} \left(\frac{b}{ax}; q\right)_{\infty} \left(\frac{q}{a}; q\right)_{\infty} \left(x; q\right)_{\infty}}.
\]
Here as usual we denote
\[
 (b; q)_{\infty} :=\prod_{i=0}^{\infty} (1-bq^i), \quad (b; q)_n :=\frac{(b; q)_{\infty}}{(bq^n; q)_{\infty}}.
\]
\end{remark}
\begin{proof}  Since all elements of each fixed Heisenberg irreducible module  have the same charge, the following holds  for any Heisenberg irreducible module $V$ in  the decomposition into irreducible modules:
\[
dim_{q, z}^{L_0, h^{\mathbb{Z}}_0} V =\frac{z^{\lambda}q^{deg (v)}}{\prod_{l=1}^{\infty} (1-q^{2l})}.
\]
From  \eqref{eqn:isoHWV-Polyn} we have
\begin{equation}\label{eqn:charproduct}
dim_{q, z}^{L_0, h^{\mathbb{Z}}_0} \mathit{F_{\chi}} = \frac{dim_{q, z}^{L_0, h^{\mathbb{Z}}_0} \mathit{F^{hwv}_{\chi}}}{\prod_{l=1}^{\infty} \big(1-q^{2l})}.
\end{equation}
The character formula \eqref{eqn:charform-ferm} will then result from \eqref{eqn:characterF-hwv}. Hence we now proceed to prove \eqref{eqn:characterF-hwv},  by  using the fermionic description of the vector space $\mathit{F^{hwv}_\chi}$. From  Corollary \ref{cor:FDecomp} we have that  $\mathit{F^{hwv}_\chi}$  is spanned by vectors of the form
\[
H^{\beta}_{(n_1)}H^{\beta}_{(n_2)}\dots H^{\beta}_{(n_k)}H^{\gamma}_{(m_1)}H^{\gamma}_{(m_2)}\dots H^{\gamma}_{(m_l)}|0\rangle,
\]
where   $n_s <0, s=1, 2, \dots k$ and $m_s <0, s=1, 2, \dots l$.
The problem is that although an action by a  $H^{\gamma}_{(m_s)}$ always raises the charge by 1, and action by a  $H^{\beta}_{(n_s)}$ always lowers the charge by 1, those actions are not so uniform with respect to the degrees, as we saw in Lemma \ref{lem:degree-calculation} and its proof.  The $H^{\beta}_{(n)}$ (and similarly $H^{\gamma}_{(n)}$) raise the degree differently depending on the charge  of the element on which they act, but in a predictable manner, via the 'resetting the counter" and the "skipping steps" principles. The formula we need has to take this into account. Thus we have to start by considering how the product type contribution will depend on  the charge of the vector it is generated from. We start from the vector $\chi _{-3/2}^n |0\rangle$, which is the minimal degree  vector of charge $-n$.   We can see from  Lemma \ref{lem:degree-calculation} that acting by all possible $H^{\beta}_{(-m)}$, $-m <-n$, on $\chi _{-3/2}^n |0\rangle$ will produce all the various factors of type $z^{-1}q^{2m-\frac{1}{2}}$, $m\geq 1$. Thus we see that for the various vectors of the type
\[
H^{\beta}_{(-m_1)}H^{\beta}_{(-m_2)}\dots H^{\beta}_{(-m_k)}\chi _{-3/2}^n|0\rangle,
\]
where   $-m_s <-n, s=1, 2, \dots k$,  the following product is a "first iteration" to this part of  the  contribution to the graded dimension:
 \begin{equation}\label{eqn:product1}
z^{-n}q^{\frac{3n}{2}}\prod_{m=1}^{\infty} \left(1+z^{-1}q^{2m-\frac{1}{2}}\right).
\end{equation}
This product is factually wrong, but "morally correct". It is factually wrong  due to the fact that in this product the contribution coming from
 $H^{\beta}_{(-n-1)}\chi _{-3/2}^n|0\rangle$
is correct ($z^{-n}q^{\frac{3n}{2}}\cdot z^{-1}q^{\frac{3}{2}}$), but the contribution corresponding to
 $H^{\beta}_{(-n-2)}H^{\beta}_{(-n-1)}\chi _{-3/2}^n|0\rangle$ is off by a factor of $q^2$: due to the "resetting the counter" principle of  Lemma \ref{lem:degree-calculation} it should be
 \[
 z^{-n}q^{\frac{3n}{2}}\cdot z^{-1}q^{\frac{3}{2}}\cdot z^{-1}q^{\frac{3}{2}},
 \]
but in the product above it is
\[
 z^{-n}q^{\frac{3n}{2}}\cdot z^{-1}q^{\frac{3}{2}}\cdot z^{-1}q^{\frac{7}{2}}.
 \]
Nevertheless, as we will see, these discrepancies  will be  to our advantage, as they  will allow us, when counting  the infinite number of  repetitions,  to count  each repetition with a different factor, and thereby avoid infinite coefficients when summing up infinitely repeated elements. The reason is that  we will have to consider all the varying products generated from the minimal degree vectors, not just a single product.  Thus we will have the same vector $\chi _{-3/2}^{n}|0\rangle$  first appear directly when generating from itself, as above, then appear again when generating  from
$\chi _{-3/2}^{n-1}|0\rangle$ as   $H^{\beta}_{(-n)}\chi _{-3/2}^{n-1}|0\rangle$, and again  from $\chi _{-3/2}^{n-2}|0\rangle$ as  $H^{\beta}_{(-n)}H^{\beta}_{(-n+1)}\chi _{-3/2}^{n-2}|0\rangle$, etc...
The products \eqref{eqn:product1}, when considering all $n\in \mathbb{N}$,  count each element  $\chi _{-3/2}^{n}|0\rangle$ twice with a coefficient 1, then once with a coefficient $q^2$, and in general, once with a coefficient $q^{2T_m}$, where $T_m$ is the $m$th triangular number.

Similarly, at each fixed charge $n>0$, the vector $\chi _{-1/2}^n |0\rangle$ is the vector of minimal degree with that charge $n$. Action by  $H^{\gamma}_{(-m)}$, $-m < -n$, on $\chi _{-1/2}^n |0\rangle$ will produce all  the various factors of type $zq^{2m-\frac{3}{2}}$.  The various vectors  of the type
\[
H^{\gamma}_{(-m_1)}H^{\gamma}_{(-m_2)}\dots H^{\gamma}_{(-m_k)}\chi _{-1/2}^n|0\rangle,
\]
where   $-m_s <-n, s=1, 2, \dots k$, will  result in the following product contributions to the graded dimension:
\begin{equation}\label{eqn:product2}
z^nq^{\frac{n}{2}}\prod_{m=1}^{\infty} \left(1+zq^{2m-\frac{3}{2}}\right),
\end{equation}
of course again taking into account that in these products the "repetitions" are counted with different factors,  not with their exact contributions (i..e, the products as above are again factually wrong, but "morally correct").

The two types of products above, \eqref{eqn:product1} and \eqref{eqn:product2} represent  the actions by the $H^{\beta}_{(n)}$ on highest weight vectors of negative charge, as well as $H^{\gamma}_{(n)}$ on highest weight vectors of poisitve charge (not in an exact, but in a "morally correct",  manner).
Now, let us consider the actions by the $H^{\beta}_{(n)}$ on highest weight vectors of positive charge, as well as $H^{\gamma}_{(n)}$ on highest weight vectors of negative charge. Due to the already heavy indexing, we  will consider  representative examples, instead of working with  generally indexed highest weight vectors. The first  example of a "mixture" $H^{\beta} (z^2)\big(\chi _{-1/2} |0\rangle\big)$  shows, see Lemma \ref{lem:degree-calculation}, that according to the "skipping steps" principle, the factor
$
(1+z^{-1}q^{\frac{3}{2}})
$
is missing in the would-be product contribution, due to the fact that $H^{\beta}_{0}$ always acts as 0 on $\mathit{F^{hwv}_\chi}$. We see that in this "mixture" we would have a product of the type (always keeping in mind that  this product is "morally correct", but factually wrong):
\begin{equation}\label{eqn:prod1}
\frac{zq^{\frac{1}{2}}\prod_{m=0}^{\infty}\left(1+z^{-1}q^{2m-\frac{1}{2}}\right)}{1+z^{-1}q^{\frac{3}{2}}}.
\end{equation}
A final example:
\begin{align*}
H^{\gamma} (z^2)\big(\chi _{-3/2}^2 |0\rangle\big)
&=H^{\gamma} (z^2)\big(H^{\beta}_{(-2)}H^{\beta}_{(-1)} |0\rangle\big) \\
&= H^{\gamma}_{(2)}\big(H^{\beta}_{(-2)}H^{\beta}_{(-1)} |0\rangle\big)\frac{1}{z^6} +H^{\gamma}_{(1)}\big(H^{\beta}_{(-2)}H^{\beta}_{(-1)} |0\rangle\big)\frac{1}{z^4}  +H^{\gamma}_{(0)}\big(\chi _{-3/2}^2 |0\rangle\big)\frac{1}{z^2} +H^{\gamma}_{(-1)}\big(\chi _{-3/2}^2 |0\rangle\big)  +\dots  \\
&=\frac{-H^{\beta}_{(-1)} |0\rangle}{z^6} + \frac{H^{\beta}_{(-2)} |0\rangle}{z^4} +0\cdot \frac{1}{z^2} + v_{\frac{15}{2};  0}  + v_{\frac{19}{2};  0}z^2 + \dots  +  \\
&= \frac{-\chi _{-3/2}|0\rangle}{z^6} + \frac{2\chi _{-7/2}|0\rangle -\chi _{-3/2}^2\chi _{-1/2}|0\rangle}{z^4} + 0\cdot \frac{1}{z^2}  +  v_{\frac{15}{2};  0} + \dots
\end{align*}
Hence action of $H^{\gamma}_{(2)}$ on $\chi _{-3/2}^2 |0\rangle$ lowers the degree by $\frac{3}{2}$,  $H^{\gamma}_{(1)}$ raises the degree by $\frac{1}{2}$,  $H^{\gamma}_{0}$ acts as 0, $H^{\gamma}_{(-1)}$  raises the degree by $\frac{9}{2}$,  $H^{\gamma}_{(-2)}$ raises the degree by $\frac{13}{2}$, and in general, acting on $\chi _{-3/2}^2 |0\rangle$,  $H^{\gamma}_{(-n)}$ raises the degree by $2n+\frac{5}{2}$. Most importantly, the factor
$
(1+zq^{\frac{5}{2}})
$
is missing, as per the "skipping steps" principle,  due to $H^{\gamma}_{0}$ always acting as 0 on $\mathit{F^{hwv}_\chi}$.

Similar calculation can be done in general for  the actions of $H^{\beta}_{(m)}$ on  $\chi _{-1/2}^n |0\rangle$; as well as for   the actions of $H^{\gamma}_{(m)}$ on  $\chi _{-3/2}^n |0\rangle$; always remembering that the factors corresponding to the actions of $H^{\beta}_{0}$ and  $H^{\gamma}_{0}$, are missing.

The next crucial consideration is that any element of $\mathit{F^{hwv}_\chi}$ can be generated  from any of the "minimal degree" vectors of any fixed charge (the $\chi _{-3/2}^n |0\rangle$ if the  charge $-n$ is negative, or $\chi _{-3/2}^2 |0\rangle$ if the charge $n$ is positive), by action of  $H^{\gamma}_{(m)}$ and  $H^{\beta}_{(m)}$, but with $m\in \mathbb{Z}$ (positive and negative indexes allowed).

  Thus, before we even try to calculate  the exact sums of (factually correct) products resulting from  the corresponding actions on each of the "minimal degree" vectors $\chi _{-1/2}^n |0\rangle$ and $\chi _{-3/2}^n |0\rangle$, we need to consider that if we  do just that we would have counted each basis vector infinitely many times, and thus such a sum would be pointless to calculate.  Counting the degree contributions  "by charge" (and therefore  infinite multiple counting)  is  unavoidable,  because  of  the  way the degree  contributions  change  depends  on the charge. Thus we do not have a single infinite product, but instead we have to sum infinitely products generated from each of the minimal degree vectors, which will entail encountering  each basis vector infinitely many times. And so instead of trying to calculate the exact products contributions (we are not sure it is possible), we will consider instead adding the "morally correct" (but factually wrong) product contributions. Not only will this  actually uniformize the products, but using  the not-factually exact, but "morally"-correct products will allow us to count each basis vector with a different coefficient each time it is encountered.  Namely,  we consider the products, for each $n\in \mathbb{Z}_{\geq 0}$:
  \begin{equation}\label{eqn:actualproducts}
  \frac{z^nq^{\frac{n}{2}} \prod_{l=1}^{\infty} \big(1+zq^{2l-\frac{3}{2}}\big)  \prod_{l=1}^{\infty} \big(1+z^{-1}q^{2l-\frac{1}{2}}\big)}{1+z^{-1}q^{2n+\frac{3}{2}}}, \quad \text{and} \quad \quad
    \frac{z^{-n}q^{\frac{3n}{2}}\prod_{l=1}^{\infty} \big(1+zq^{2l-\frac{3}{2}}\big)  \prod_{l=1}^{\infty} \big(1+z^{-1}q^{2l-\frac{1}{2}}\big) }{1+zq^{2n+\frac{1}{2}}}
  \end{equation}
 Observe in the above that we  "shifted",   for instance  compare  the  first product above, for $n=1$, vs the product of  \eqref{eqn:prod1}--- we shifted the "missing factor" from $(1+z^{-1}q^{\frac{3}{2}})$  in \eqref{eqn:prod1} to $(1+z^{-1}q^{\frac{7}{2}})$ in \eqref{eqn:actualproducts}, as well as the starting factor from $(1+z^{-1}q^{-\frac{1}{2}})$ to $(1+z^{-1}q^{\frac{3}{2}})$. Similarly, we
  know a factor would be missing from $\prod_{l=1}^{\infty} \big(1+z^{-1}q^{2l-\frac{1}{2}}\big)$ due to  $H^{\beta}_{(m)}$ always acting by 0, but we shifted  all the factors from  $1+z^{-1}q^{2m-\frac{1}{2}}$  in \eqref{eqn:prod1}  to $1+z^{-1}q^{2m+\frac{3}{2}}$ in the product in \eqref{eqn:actualproducts}. For example, this  will entail that the minimal degree vector $\chi _{-3/2} |0\rangle =H^{\beta}_{(-1)}|0\rangle$, which can be obtained by the consecutive action $H^{\beta}_{(-1)}H^{\beta}_{(1)}\chi _{-3/2}|0\rangle$, is counted with a different factor of $q^{2+4}=q^{2T_2}$ in the first product  in \eqref{eqn:actualproducts}. Thus the shifts are introduced to allow for the "skipping the steps" principle, but also to make sure that we don't sum the same exact term (corresponding to the same exact basis vector)  infinitely many times, instead we sum it with a different degree coefficient each time.

Thus, generalizing the observation from above, we see that if we  sum the products \eqref{eqn:actualproducts} we would have counted each basis vector infinitely many times, but  with a different degree coefficient each time. That coefficient is $q^{2T_m}$, where $T_m$ is the $m$-th triangular number--- in fact we encounter the  coefficient  $q^{2T_m}$ exactly twice among the infinite number of times we encounter  each basis vector. Thus the sum of all the products from  \eqref{eqn:actualproducts} equals
\begin{equation}
\left(2\sum_{m\in \mathbb{Z}_{\geq 0}}q^{2T_m}\right)\cdot dim_{q, z}^{L_0, h^{\mathbb{Z}}_0} \mathit{F^{hwv}_{\chi}},
\end{equation}
and so
\begin{equation}
\left(2\sum_{m\in \mathbb{Z}_{\geq 0}}q^{2T_m}\right)\cdot dim_{q, z}^{L_0, h^{\mathbb{Z}}_0} \mathit{F^{hwv}_{\chi}} =\prod_{l=1}^{\infty} \big(1+zq^{2l-\frac{3}{2}}\big)  \prod_{l=1}^{\infty} \big(1+z^{-1}q^{2l-\frac{1}{2}}\big)\cdot \left(\sum_{n=0}^{\infty} \frac{z^nq^{\frac{n}{2}} }{1+z^{-1}q^{2n+\frac{3}{2}}} + \sum_{n=0}^{\infty} \frac{z^{-n}q^{\frac{3n}{2}} }{1+zq^{2n+\frac{1}{2}}}\right).
\end{equation}
Now we use the following Jacobi formula for the triangular numbers \footnote{This identity  can be found on page 185 of the original manuscript by Jacobi, \cite{Jacobi1829}.}
\[
\sum_{m\in \mathbb{Z}_{\geq 0}}q^{2T_m} = 1+q^2 +q^4 +q^{12} +q^{20} +\dots + q^{2T_m} +\dots = \frac{\prod_{i=1}^{\infty} (1-q^{4i})}{\prod_{i=1}^{\infty} (1-q^{4i-2})} =\prod_{i=1}^{\infty} (1-q^{4i})(1+q^{2i}).
\]
Hence we get
\begin{equation}
dim_{q, z}^{L_0, h^{\mathbb{Z}}_0} \mathit{F^{hwv}_{\chi}} =\frac{\prod_{l=1}^{\infty} \big(1+zq^{2l-\frac{3}{2}}\big)  \prod_{l=1}^{\infty} \big(1+z^{-1}q^{2l-\frac{1}{2}}\big)}{2 \prod_{i=1}^{\infty} (1-q^{4i})(1+q^{2i})}\cdot \left(\sum_{n=0}^{\infty} \frac{z^nq^{\frac{n}{2}} }{1+z^{-1}q^{2n+\frac{3}{2}}} + \sum_{n=0}^{\infty} \frac{z^{-n}q^{\frac{3n}{2}} }{1+zq^{2n+\frac{1}{2}}}\right),
\end{equation}
which is precisely \eqref{eqn:characterF-hwv}.
Thus  using \eqref{eqn:charproduct} we derive \eqref{eqn:charform-ferm}.
\end{proof}
\begin{cor}\label{cor:threecharacters}
The following equalities hold:
\begin{align*}
dim_{q, z}^{L_0, h^{\mathbb{Z}}_0} \mathit{F^{hwv}_{\chi}}  & =  \sum_{(\pi_1  | \pi_2) \in \mathfrak{BP}_{DI}} q^{\mathcal{W}\left((\pi_1 | \pi_2)\right)}z^{birank\left((\pi_1 | \pi_2)\right)}  =\prod_{n\geq 1} {\frac{(1-q^{2n})}{ \big(1-zq^{2n- \frac{3}{2}}\big) \big(1-z^{-1}q^{2n- \frac{1}{2}}\big)}}\\
&=\prod_{l=1}^{\infty} {\frac{\big(1+zq^{2l-\frac{3}{2}}\big) \big(1+z^{-1}q^{2l-\frac{1}{2}}\big)}{2 (1-q^{4l})(1+q^{2l})}}\cdot \left(\sum_{k=0}^{\infty} \frac{z^kq^{\frac{k}{2}} }{1+z^{-1}q^{2k+\frac{3}{2}}} + \sum_{k=0}^{\infty} \frac{z^{-k}q^{\frac{3k}{2}} }{1+zq^{2k+\frac{1}{2}}}\right).
\end{align*}
\end{cor}
Finally, we would like to underline a  connection between the CKP hierarchy, its bosonization (and thus the $\beta-\gamma$ system), and Dyson's crank of a partition. Recall the Dyson crank of a partition $\lambda$ is defined as follows (\cite{Garvan-crank}): Let $l(\lambda)$ denote the largest part of $\lambda$, $\omega (\lambda)$ denote the number of 1's in $\lambda$, and $\mu(\lambda)$  denote the number of parts of $\lambda$ larger than $\omega (\lambda)$. The crank $c(\lambda)$  is given by
\begin{align*}
     {\displaystyle c(\lambda )={\begin{cases}l(\lambda )&{\text{ if }}\omega (\lambda )=0\\\mu (\lambda )-\omega (\lambda )&{\text{ if }}\omega (\lambda )>0.\end{cases}}}
     \end{align*}

Denote by $N^\prime(m,n)$  the number of partitions of $n$ with crank equal to $m$,  with the exception that for   $n=1$ we set $N^\prime(-1,1) = -N^\prime(0,1) = N^\prime(1,1) = 1$. For convenience we also set $N^\prime(m,n)=0$ whenever $m$ or $n$ is not an integer.
Specifically, the generating function for $N^\prime(m,n)$ is (\cite{Garvan-crank}):
\begin{equation}\label{eqn:crank}
{\displaystyle \sum _{n=0}^{\infty }\sum _{m=-\infty }^{\infty }N^\prime(m,n)z^{m}q^{n}=\prod _{n=1}^{\infty }{\frac {(1-q^{n})}{(1-zq^{n})(1-z^{-1}q^{n})}}}.
\end{equation}
Comparing this generating function  to the identities of Corollary \ref{cor:threecharacters}, we see that:
\begin{cor}\label{cor:dysoncrankcharacter}
\begin{align}
 dim_{q, z}^{L_0, h^{\mathbb{Z}}_0} \mathit{F^{hwv}_{\chi}}  = \sum _{n=0}^{\infty} \sum _{m=-\infty }^{\infty }&\sum _{l=0}^{\infty } N^\prime(m-l,n)z^{m}q^{2n+\frac{m}{2}} = \sum_{(\pi_1  | \pi_2) \in \mathfrak{BP}_{DI}} q^{\mathcal{W}\left((\pi_1 | \pi_2)\right)}z^{birank\left((\pi_1 | \pi_2)\right)} \\
=\prod_{n\geq 1} {\frac{(1-q^{2n})}{ \big(1-zq^{2n- \frac{3}{2}}\big) \big(1-z^{-1}q^{2n- \frac{1}{2}}\big)}}
 &= \prod_{l=1}^{\infty} {\frac{\big(1+zq^{2l-\frac{3}{2}}\big) \big(1+z^{-1}q^{2l-\frac{1}{2}}\big)}{2 (1-q^{4l})(1+q^{2l})}}\cdot \left(\sum_{k=0}^{\infty} \frac{z^kq^{\frac{k}{2}} }{1+z^{-1}q^{2k+\frac{3}{2}}} + \sum_{k=0}^{\infty} \frac{z^{-k}q^{\frac{3k}{2}} }{1+zq^{2k+\frac{1}{2}}}\right).\end{align}
\end{cor}
\begin{proof}
Note that although there are several apparently infinite  sums in $\sum _{n=0}^{\infty }\sum _{m=-\infty }^{\infty }\sum _{l=0}^{\infty } N^\prime(m-l,n)z^{m}q^{2n+\frac{m}{2}}$, such as  from $m=-\infty$ to $m=\infty$, in fact at each $n$ only finitely many of the $N^\prime(m,n)$ are nonzero. We obtain the first equality  by setting $q\to q^2$ and $z\to zq^{\frac{1}{2}}$ in \eqref{eqn:crank}:
\[
\sum _{n=0}^{\infty }\sum _{m=-\infty }^{\infty }N^\prime(m,n) z^{m}q^{2n+\frac{m}{2}}=\prod _{n=1}^{\infty }{\frac {(1-q^{2n})}{(1-zq^{2n+\frac{1}{2}})(1-z^{-1}q^{2n-\frac{1}{2}})}}.
\]
Thus we have
\[
\frac{1}{1- zq^{\frac{1}{2}}}\left(\sum _{n=0}^{\infty }\sum _{m=-\infty }^{\infty }N^\prime(m,n)z^{m}q^{2n+\frac{m}{2}}\right) = \prod _{n=1}^{\infty }{ \frac{ (1-q^{2n})}{ \big(1-zq^{2n- \frac{3}{2}}\big) \big(1-z^{-1}q^{2n- \frac{1}{2}}\big)}}.
\]
We now expand and re-sum (assuming $|zq^{\frac{1}{2}}|<1$):
\begin{align*}
\frac{1}{1- zq^{\frac{1}{2}}}\left(\sum _{n=0}^{\infty }\sum _{m=-\infty }^{\infty }N^\prime(m,n)z^{m}q^{2n+\frac{m}{2}}\right) &= \left(\sum _{l=0}^{\infty } z^l q^{\frac{l}{2}}\right)\left(\sum _{n=0}^{\infty }\sum _{m=-\infty }^{\infty }N^\prime(m,n)z^{m}q^{2n+\frac{m}{2}}\right)\\
 =\sum _{l=0}^{\infty }\sum _{n=0}^{\infty }\sum _{m=-\infty }^{\infty }N^\prime(m,n)z^{m+l}q^{2n+\frac{m+l}{2}} &=\sum _{n=0}^{\infty }\sum _{m=-\infty }^{\infty }\sum _{l=0}^{\infty } N^\prime(m-l,n)z^{m}q^{2n+\frac{m}{2}}.
\end{align*}
\end{proof}
In \cite{AngCKPSecond} we asked the following questions: It would be interesting to derive a formula giving a correspondence between a partition $\mathfrak{p}\in \mathfrak{P}_{tdo}$ of weight $n$ and the highest weight vector corresponding to that partition, or even the charge of that highest weight vector. As the weights of the partitions grow, the charges are less straightforward to calculate.  For example at weight $\frac{13}{2}$ there are 7 partitions from $\mathfrak{P}_{tdo}$ and one can calculate by brute force that there is a highest weight vector of charge 13, a highest weight vector of charge 9, two highest weight vectors of charge 5,  two highest weight vectors of charge 1 and  a highest weight vector of charge $-3$.

We cannot yet give a direct formula between a $\mathfrak{p}\in \mathfrak{P}_{tdo}$ of weight $n$ and the highest weight vector corresponding to that partition $\mathfrak{p}$, although we have an indirect correspondence in Proposition \ref{prop:counting-formula}. But we can answer the question about the charges of the highest weight vectors with degree $n$:
\begin{cor}
The number of highest weight vectors with degree $n$ and charge $m$ equals $\sum_{l=0}^{\infty} N^\prime(m-l, \frac{2n-m}{4})$.
\end{cor}
\begin{remark}
We would also note that the above formula confirms that  $2deg(v)\equiv chg(v) \ \ (mod \ 4)$, where $deg(v)$ is  the degree of the  highest weight vector $v$, and $chg(v)$ is its charge.
\end{remark}
We would like to thank Thomas Creutzig for the helpful discussion, in particular he confirmed that such a connection with the Dyson  crank should hold, and explained the physicist concept of a coset space. We are also grateful to  Kailash Misra and Naihuan Jing, for organizing  this AMS Special Session, and  for their invitation.

\section{Appendix}
In this Appendix we give another  derivation of  the identity \eqref{eqn:IdentityR} from the Ramanujan summation formula (see e.g. \cite{hardy1940ramanujan},  \cite{AndrewsRam-1},  \cite{AndrewsRam})   by appropriate changes of variables. It serves as a confirmation of the validity of this  rather curious identity.

We start with the Ramanujan Psi summation formula
\[
_1\psi_1 (a, b; q, x)  = \sum_{k=-\infty}^{\infty} \frac{(a; q)_k}{(b; q)_k}x^k = \frac{\left(q; q\right)_{\infty}\left(\frac{b}{a}; q\right)_{\infty}\left(\frac{q}{ax}; q\right)_{\infty}\left(ax; q\right)_{\infty}}{\left(b; q\right)_{\infty} \left(\frac{b}{ax}; q\right)_{\infty} \left(\frac{q}{a}; q\right)_{\infty} \left(x; q\right)_{\infty}}.
\]
The specialization we need is $b= aq$, which reduces the Ramanujan summation formula to
\[
(1-a)\sum_{k=-\infty}^{\infty} \frac{x^k}{1-aq^k} = \frac{\left(q; q\right)_{\infty}^2\left(\frac{q}{ax}; q\right)_{\infty}\left(ax; q\right)_{\infty}}{\left(aq; q\right)_{\infty} \left(\frac{q}{x}; q\right)_{\infty} \left(\frac{q}{a}; q\right)_{\infty} \left(x; q\right)_{\infty}}.
\]
We now replace $x=qy$  and we get
\[
\frac{1}{1-a}+\sum_{k=1}^{\infty} \frac{y^kq^k}{1-aq^k}  -\frac{1}{a}\sum_{k=1}^{\infty} \frac{y^{-k}}{1-a^{-1}q^k} = \frac{1}{(1-a)}\frac{\left(q; q\right)_{\infty}^2\left((ay)^{-1}; q\right)_{\infty}\left(qay; q\right)_{\infty}}{\left(aq; q\right)_{\infty} \left(y^{-1}; q\right)_{\infty} \left(\frac{q}{a}; q\right)_{\infty} \left(qy; q\right)_{\infty}}.
\]
Now set $q\mapsto q^2$:
\[
\frac{1}{1-a}+\sum_{k=1}^{\infty} \frac{y^kq^{2k}}{1-aq^{2k}}  -\frac{1}{a}\sum_{k=1}^{\infty} \frac{y^{-k}}{1-a^{-1}q^{2k}} = \frac{1}{(1-a)}\frac{\left(q^2; q^2\right)_{\infty}^2\left((ay)^{-1}; q^2\right)_{\infty}\left(q^2ay; q^2\right)_{\infty}}{\left(aq^2; q^2\right)_{\infty} \left(y^{-1}; q^2\right)_{\infty} \left(\frac{q^2}{a}; q^2\right)_{\infty} \left(q^2y; q^2\right)_{\infty}},
\]
followed by  $y =zq^{\frac{-3}{2}}$:
\[
\frac{1}{1-a}+\sum_{k=1}^{\infty} \frac{z^k q^{\frac{k}{2}}}{1-aq^{2k}}  -\frac{1}{a}\sum_{k=1}^{\infty} \frac{z^{-k}q^{\frac{3k}{2}}}{1-a^{-1}q^{2k}} = \frac{1}{(1-a)}\frac{\left(q^2; q^2\right)_{\infty}^2\left(a^{-1}z^{-1}q^{\frac{3}{2}}; q^2\right)_{\infty}\left(azq^{\frac{1}{2}}; q^2\right)_{\infty}}{\left(aq^2; q^2\right)_{\infty} \left(z^{-1}q^{\frac{3}{2}}; q^2\right)_{\infty} \left(\frac{q^2}{a}; q^2\right)_{\infty} \left(zq^{\frac{1}{2}}; q^2\right)_{\infty}}.
\]
Finally, letting $a=-z^{-1}q^{\frac{3}{2}}$ we get
\begin{align*}
\frac{1}{1+z^{-1}q^{\frac{3}{2}}}+\sum_{k=1}^{\infty} \frac{z^k q^{\frac{k}{2}}}{1+z^{-1}q^{2k+\frac{3}{2}}}  & +\sum_{k=1}^{\infty} \frac{z^{-k+1}q^{\frac{3k-3}{2}}}{1+zq^{2k-\frac{3}{2}}}\\
 &= \frac{1}{(1+z^{-1}q^{\frac{3}{2}})}\frac{\left(q^2; q^2\right)_{\infty}^2\left(-1; q^2\right)_{\infty}\left(q^2; q^2\right)_{\infty}}{\left(-z^{-1}q^{\frac{7}{2}}; q^2\right)_{\infty} \left(z^{-1}q^{\frac{3}{2}}; q^2\right)_{\infty} \left(-zq^{\frac{1}{2}}; q^2\right)_{\infty} \left(zq^{\frac{1}{2}}; q^2\right)_{\infty}}.
\end{align*}
Which simplifies to
\begin{align*}
\sum_{k=0}^{\infty} \frac{z^k q^{\frac{k}{2}}}{1+z^{-1}q^{2k+\frac{3}{2}}} & +\sum_{k=0}^{\infty} \frac{z^{-k}q^{\frac{3k}{2}}}{1+zq^{2k+\frac{1}{2}}}\\
 &= \frac{2\prod_{n=1}^{\infty} \left(1-q^{2n}\right)^2 \prod_{n=1}^{\infty} \left(1+q^{2n}\right)^2}{\prod_{n=1}^{\infty} \left(1+z^{-1}q^{2n-\frac{1}{2}}\right)\prod_{n=1}^{\infty} \left(1-z^{-1}q^{2n-\frac{1}{2}}\right)\prod_{n=1}^{\infty} \left(1-zq^{2n-\frac{3}{2}}\right)\prod_{n=1}^{\infty} \left(1+zq^{2n-\frac{3}{2}}\right)}.
\end{align*}

\def\cprime{$'$}


\begin{thebibliography}{DJKM81b}

\bibitem[AA78]{AndrewsRam}
George~E. Andrews and Richard Askey.
\newblock A simple proof of Ramanujan's summation of the $_1\psi_1$.
\newblock {\em Aequationes Mathematicae}, 18(1):333--337, 1978.

\bibitem[Abe07]{Abe}
Toshiyuki Abe.
\newblock A $\mathbb{Z}_2$-orbifold model of the symplectic fermionic vertex
  operator superalgebra.
\newblock {\em Mathematische Zeitschrift}, 255(4):755--792, 2007.

\bibitem[ACJ14]{ACJ}
Iana~I. Anguelova, Ben Cox, and Elizabeth Jurisich.
\newblock ${N}$-point locality for vertex operators: normal ordered products,
  operator product expansions, twisted vertex algebras.
\newblock {\em J. Pure Appl. Algebra}, 218(12):2165--2203, 2014.

\bibitem[AG88]{Garvan-crank}
George~E. Andrews and F.~G. Garvan.
\newblock Dyson's crank of a partition.
\newblock {\em Bull. Amer. Math. Soc. (N.S.)}, 18(2):167--171, 04 1988.

\bibitem[{And}69]{AndrewsRam-1}
G.E. {Andrews}.
\newblock {On {R}amanujan's summation of $\sb 1\psi\sb 1 (a;b;z)$.}
\newblock {\em {Proc. Am. Math. Soc.}}, 22:552--553, 1969.

\bibitem[Ang13a]{Ang-Varna2}
Iana~I. Anguelova.
\newblock Boson-fermion correspondence of type {B} and twisted vertex algebras.
\newblock In Vladimir Dobrev, editor, {\em Lie Theory and Its Applications in
  Physics}, volume~36 of {\em Springer Proceedings in Mathematics and
  Statistics}, pages 399--410. Springer Japan, 2013.

\bibitem[Ang13b]{AngTVA}
Iana~I. Anguelova.
\newblock Twisted vertex algebras, bicharacter construction and boson-fermion
  correspondences.
\newblock {\em Journal of Mathematical Physics}, 54(12):38, 2013.

\bibitem[Ang14]{AngD-A}
Iana~I. Anguelova.
\newblock Boson-fermion correspondence of type {D}-{A} and multi-local
  {V}irasoro representations on the {F}ock space $\mathit{F^{\otimes
  \frac{1}{2}}}$.
\newblock {\em Journal of Mathematical Physics}, 55(11):23, 2014.

\bibitem[Ang15]{AngMLB}
Iana~I. Anguelova.
\newblock Multilocal bosonization.
\newblock {\em Journal of Mathematical Physics}, 56(12):13, 2015.

\bibitem[Ang17]{AngCKPSecond}
Iana~I. Anguelova.
\newblock The second bosonization of the {CKP} hierarchy.
\newblock {\em Journal of Mathematical Physics}, 58(7):071707, 2017.

\bibitem[CR13]{Creutzig-Ridout}
Thomas Creutzig and David Ridout.
\newblock Relating the archetypes of logarithmic conformal field theory.
\newblock {\em Nuclear Physics B}, 872(3):348 -- 391, 2013.

\bibitem[DJKM81a]{DJKM-KP}
Etsur{\=o} Date, Michio Jimbo, Masaki Kashiwara, and Tetsuji Miwa.
\newblock Transformation groups for soliton equations. {III}. {O}perator
  approach to the {K}adomtsev-{P}etviashvili equation.
\newblock {\em J. Phys. Soc. Japan}, 50(11):3806--3812, 1981.

\bibitem[DJKM81b]{DJKM6}
Etsur{\=o} Date, Michio Jimbo, Masaki Kashiwara, and Tetsuji Miwa.
\newblock Transformation groups for soliton equations. {VI}. {KP} hierarchies
  of orthogonal and symplectic type.
\newblock {\em J. Phys. Soc. Japan}, 50(11):3813--3818, 1981.

\bibitem[DJKM82]{DJKM-4}
Etsur{\=o} Date, Michio Jimbo, Masaki Kashiwara, and Tetsuji Miwa.
\newblock Transformation groups for soliton equations. {IV}. {A} new hierarchy
  of soliton equations of {KP}-type.
\newblock {\em Phys. D}, 4(3):343--365, 1981/82.

\bibitem[FLM88]{FLM}
Igor Frenkel, James Lepowsky, and Arne Meurman.
\newblock {\em Vertex operator algebras and the {M}onster}, volume 134 of {\em
  Pure and Applied Mathematics}.
\newblock Academic Press Inc., Boston, MA, 1988.

\bibitem[FMS86]{FMS}
Daniel Friedan, Emil Martinec, and Stephen Shenker.
\newblock Conformal invariance, supersymmetry and string theory.
\newblock {\em Nuclear Phys. B}, 271(1):93--165, 1986.

\bibitem[Gar10]{Garvan}
F.~G Garvan.
\newblock Biranks for partitions into 2 colors.
\newblock In {\em Ramanujan rediscovered, volume 14 of Ramanujan Math. Soc.
  Lect. Notes Ser.}, pages 87--111. Ramanujan Math. Soc., Mysore, 2010.

\bibitem[HA40]{hardy1940ramanujan}
G.H. Hardy and S.R. Aiyangar.
\newblock {\em Ramanujan: Twelve Lectures on Subjects Suggested by His Life and
  Work}.
\newblock Cambridge University Press, 1940.

\bibitem[Hir04]{Hirota}
Ryogo Hirota.
\newblock {\em The Direct Method in Soliton Theory}.
\newblock Cambridge University Press, 2004.

\bibitem[HL04]{HL}
P.~Hammond and R.~Lewis.
\newblock Congruences in ordered pairs of partitions.
\newblock {\em Int. J. Math. Math.}, 2004(47):2509--2512, 2004.

\bibitem[Jac29]{Jacobi1829}
C.G.J. Jacobi.
\newblock {\em Fundamenta nova theoriae functionum ellipticarum}.
\newblock Sumtibus fratrum, 1829.

\bibitem[Kac90]{Kac-Lie}
Victor~G. Kac.
\newblock {\em Infinite-dimensional {L}ie algebras}.
\newblock Cambridge University Press, Cambridge, third edition, 1990.

\bibitem[Kac98]{Kac}
Victor Kac.
\newblock {\em Vertex algebras for beginners}, volume~10 of {\em University
  Lecture Series}.
\newblock American Mathematical Society, Providence, RI, second edition, 1998.

\bibitem[KR87]{KacRaina}
V.~G. Kac and A.~K. Raina.
\newblock {\em Bombay lectures on highest weight representations of
  infinite-dimensional {L}ie algebras}, volume~2 of {\em Advanced Series in
  Mathematical Physics}.
\newblock World Scientific Publishing Co. Inc., Teaneck, NJ, 1987.

\bibitem[KWY98]{WangKac}
Victor~G. Kac, Weiqiang Wang, and Catherine~H. Yan.
\newblock Quasifinite representations of classical {L}ie subalgebras of {$
  W_{1+\infty}$}.
\newblock {\em Adv. Math.}, 139(1):56--140, 1998.


\bibitem[LL04]{LiLep}
James Lepowsky and Haisheng Li.
\newblock {\em Introduction to vertex operator algebras and their
  representations}, volume 227 of {\em Progress in Mathematics}.
\newblock Birkh\"auser Boston Inc., Boston, MA, 2004.


\bibitem[MJD00]{Miwa-book}
T.~Miwa, M.~Jimbo, and E.~Date.
\newblock {\em Solitons: differential equations, symmetries and infinite
  dimensional algebras}.
\newblock Cambridge tracts in mathematics. Cambridge University Press, 2000.

\bibitem[Rid10]{Ridout}
David Ridout.
\newblock $\widehat{sl}(2)_{-1/2}$ and the triplet model.
\newblock {\em Nuclear Physics B}, 835(3):314 -- 342, 2010.

\bibitem[vOS12]{OrlovLeur}
J.~W. {van de Leur}, A.~Y. {Orlov}, and T.~{Shiota}.
\newblock {CKP Hierarchy, Bosonic Tau Function and Bosonization Formulae}.
\newblock {\em SIGMA}, 8, 2012.
\newblock 28pp.

\bibitem[Wan99]{WangDuality}
Weiqiang Wang.
\newblock Duality in infinite-dimensional {F}ock representations.
\newblock {\em Commun. Contemp. Math.}, 1(2):155--199, 1999.

\bibitem[You89]{YouBKP}
Yuching You.
\newblock Polynomial solutions of the {BKP} hierarchy and projective
  representations of symmetric groups.
\newblock In {\em Infinite-dimensional {L}ie algebras and groups
  ({L}uminy-{M}arseille, 1988)}, volume~7 of {\em Adv. Ser. Math. Phys.}, pages
  449--464. World Sci. Publ., Teaneck, NJ, 1989.

\end{thebibliography}
\end{document}